\documentclass{article}

\PassOptionsToPackage{numbers}{natbib}

\usepackage{tabularx}
\usepackage{arxiv}
\usepackage{diagbox}
\usepackage{wrapfig}
\usepackage{placeins}

\usepackage[utf8]{inputenc} %
\usepackage[T1]{fontenc}    %
\usepackage{hyperref}       %
\usepackage{url}            %
\usepackage{booktabs}       %
\usepackage{amsfonts}       %
\usepackage{nicefrac}       %
\usepackage{microtype}      %
\usepackage{xcolor}         %
\usepackage{amsmath}
\usepackage{amsthm}
\usepackage{amssymb}
\usepackage{cleveref}
\usepackage{tikz}
\usepackage{xcolor}
\usepackage{graphicx}
\usepackage{listings}
\usetikzlibrary{arrows.meta,positioning,fit,calc,backgrounds}
\newtheorem{theorem}{Theorem}

\newtheorem{definition}{Definition}

\newcommand{\Anc}{\operatorname{Anc}}
\newcommand{\Dep}{\operatorname{Dep}}

\newcommand{\accept}{\mathsf{accept}}
\newcommand{\reject}{\mathsf{reject}}

\lstdefinestyle{prompt}{
    basicstyle=\ttfamily\footnotesize,
    breaklines=true,
    breakatwhitespace=false,
    columns=fullflexible,
    keepspaces=true,
    showstringspaces=false,
    frame=single,
    framerule=0.4pt,
    rulecolor=\color{gray!50},
    backgroundcolor=\color{gray!5},
    aboveskip=0.6em,
    belowskip=0.6em,
    xleftmargin=0.6em,
    xrightmargin=0.6em,
    literate={–}{{-}}1 {—}{{---}}1 {…}{{...}}1,
}
\lstnewenvironment{promptbox}{\lstset{style=prompt}}{}
\title{Pseudo-Formalization for Automatic Proof Verification}

\author{Slim Barkallah*, Luke Bailey*, Kaiyue Wen*, Mohammed Abouzaid, Tengyu Ma}

\begin{document}

\maketitle

\begin{center}
    \vspace{-2mm}
   Stanford University  
\end{center}

\begin{abstract}
Reliable verification of proofs remains a bottleneck for training and evaluating AI systems on hard mathematical reasoning. Fully formal proofs, in languages like Lean, are easy to verify because they are unambiguous and modular. Most proofs, particularly those written by AI systems, have neither property, and translating them into formal languages remains challenging in many frontier math settings.
We propose Pseudo-Formalization (PF),
a proof format that captures the modularity and precision
of formal proofs while retaining the
flexibility of natural language. A Pseudo-Formal
proof is decomposed into self-contained modules, each
stating its premises, conclusion, and proof
in natural language. To verify the
correctness of a natural language proof, an LLM
translates it to Pseudo-Formal and then verifies each module
independently, an algorithm we call Block Verification (BV).
We evaluate PF+BV on three benchmarks spanning olympiad
and research-level
mathematics, where it pareto-dominates LLM-as-judge 
baselines on error-finding precision and recall.
To support future work, we release two research-level
proof verification benchmark \texttt{ArxivMathGradingBench} and 
\texttt{FirstProofGradingBench}.
\end{abstract}

{
\renewcommand{\thefootnote}{}
\makeatletter\long\def\@makefntext#1{\noindent#1}\makeatother
\footnotetext{%
* Equal Contribution\quad\quad\quad
Correspondence to \{slimbark, ljbailey, kaiyuew, abouzaid, tengyuma\}@stanford.edu \\
 Project page available at \url{https://pseudoformalization.github.io} \\
 Code available at \url{https://github.com/Slim205/pseudo-formalization} \\
 Benchmarks available at \url{https://huggingface.co/datasets/LukeBailey181Pub/}}
\addtocounter{footnote}{-1}
}

\begin{figure}[ht!]
    \centering
    \includegraphics[width=0.9\textwidth,height=189pt,keepaspectratio=false]{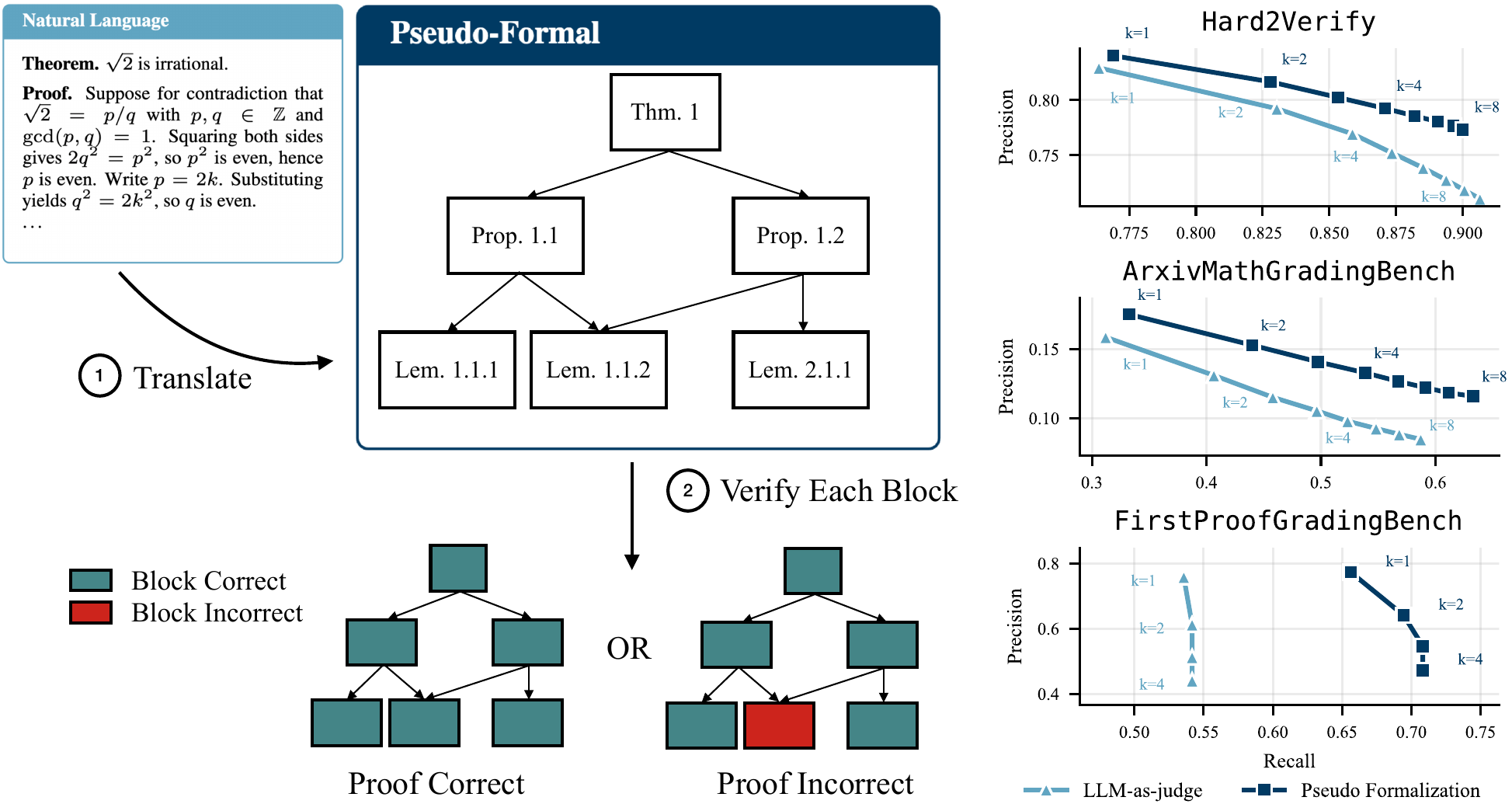}
    \caption{
        \textbf{Left:} Diagram indicating how Pseudo-Formalization
        can be used to verify proofs. We translate a natural
        language proof to Pseudo-Formal representation
        (\Cref{def:pf-proof}), and then verify each block 
        (see \Cref{fig:imo-pf-tree-gb0813} for example).
        \textbf{Right:} Results comparing Pseudo-Formalization
        to baseline LLM-as-judge verification. We report 
      the precision vs recall on identifying the mistaken step on \texttt{Hard2Verify}
        benchmark \citep{pandit2025hard2verify}, real arXiv
        papers with known errors, and first proof second batch submissions
        \citep{firstproof2}.
          Each curve traces multiple independent
          runs aggregated. Pseudo-Formalization pareto-dominates
          LLM-as-judge in all 
          settings, achieving higher recall at every precision level.}
    \label{fig:figure_1}
\end{figure}

\section{Introduction}

The ability to use AI to automatically and reliably verify the correctness
of math proofs would have substantial practical and scientific value \citep{anil2021learning,kirchner2024prover,gao2025mathminos}.
It would provide trustworthy reward signals for training LLMs on
hard mathematical reasoning, assist mathematicians in reviewing
proofs authored by humans or machines, and enable scalable evaluation
of AI mathematical capabilities \citep{lightman2024lets,zhang2025generative,liu2025trust,wang2025solvingverifying}. Yet despite progress in
proof generation \citep{lin2025goedel,lin2025goedelv2,shao2025deepseekmath,huang2025winning},
verification remains a bottleneck
\citep{luong2025towards,pandit2025hard2verify,huang2025pessimistic,mahdavi2025scaling,processbench}.
For example,  \citet{luong2025towards} find that although Gemini
Deep Think achieves a Gold medal in IMO 2025, it only gets
50\% accuracy when grading IMO proofs.

What makes a proof easy to verify? To answer this, consider the easiest
proofs to verify: ones written in fully formal languages such as Lean or
Isabelle, which can be checked mechanically by a compiler \cite{de2015lean, moura2021lean,nipkow2002isabelle}.
Formal proofs admit such easy verification for two reasons.
The first is \emph{semantic precision}: each step has a single,
precise interpretation.
This comes from the language's programmatic nature, where every symbol and inference rule is governed by a strict type system and formal grammar \cite{harrison2008formal}.
The second is \emph{modularization}: the proof
is decomposed into self-contained units
(lemmas, propositions, claims) each with explicitly stated
premises and conclusions, such that each unit can be verified
independently. Without modularization, verification
requires holding the entire argument in working memory
and tracking long chains of dependencies, which
scales poorly with proof length. Modularization
converts a single global verification problem into
many local ones, each dramatically easier: the verifier
only needs to check that the stated premises imply the
stated conclusion, taking upstream modules as given.
It also localizes errors, so a flaw in one module can be
targeted for revision without rewriting the whole proof.

Many proofs written by humans and AI systems do not obviously
satisfy these properties, making them hard to verify.
We cannot easily translate them into fully
formal languages as formal libraries needed to express
modern research mathematics are far from complete,
and even with substantial AI assistance, formalizing
mathematics research can take months or years of expert
effort \citep{buzzard2024flt,scholze2022liquid}.

We instead propose \emph{Pseudo-Formalization} (PF), which aims
to retain the modularity and precision that make formal
proofs verifiable while preserving the flexibility of natural
language, such as the ability to cite
existing results, appeal to standard techniques, and rely
on the rich vocabulary mathematicians already use. A
Pseudo-Formal proof is decomposed into modules,
each stating its premises, conclusion, and proof of the
conclusion in natural language (see \Cref{fig:figure_1}). Given a proof, we take three
steps to verify it. First, an LLM translates it to PF. 
Then, the same model verifies each module, and a calibration model aggregates the errors into a verdict under a configurable strictness threshold, an algorithm we call Block Verification (BV). We refer 
to the entire verification pipeline as PF+BV.

Pseudo-Formalization can be viewed as a relaxation of
autoformalization, where a natural language proof
is translated to a fully formal language and then checked for correctness.
Both approaches use an LLM to translate
a proof into a format that is easier to verify. Autoformalization
targets a formal language like Lean, paired with a compiler. PF instead
targets a structured natural-language representation, paired with an LLM
running Block Verification. Thus the Block Verification LLM plays
an analogous role to a compiler for a fully formal language.

Evaluating LLM math verifiers requires realistic and challenging benchmarks.
We test our method on three benchmarks covering different mathematical domains,
two of which we introduce.
The first is 
\texttt{Hard2Verify}~\citep{pandit2025hard2verify}, a benchmark 
of Olympiad and Putnam level competition-math problems paired
with AI generated solutions and step-level correctness labels. 
The second
benchmark is built from full proofs drawn
from arXiv math papers that contained errors later corrected
by the authors, providing a realistic test of verification
on research mathematics. We refer to it as
\texttt{ArxivMathGradingBench}. 
The third, \texttt{FirstProofGradingBench}, we build by converting
the research-level AI proof attempts and expert referee reports
released in the First Proof project's second batch~\citep{firstproof2}
into a verification benchmark.
Across 
all benchmarks, we find that PF+BV
pareto-dominates the baseline 
of evaluating the original natural language proof. 

Overall, our contributions can be summarized as follows:

\begin{enumerate}
    \item We introduce Pseudo-Formalization as an easy to verify proof format, and Pseudo-Formalization + Block-Verification (PF+BV) as 
    an automatic proof verification algorithm.
    \item We evaluate PF+BV on two benchmarks covering Olympiad and research-level mathematics, where it pareto-dominates strong LLM-as-judge baselines on error-finding precision and recall.
    \item 
 We release \texttt{ArxivMathGradingBench}, a new benchmark of 100 mathematics research papers with 109 known errors,
  and \texttt{FirstProofGradingBench}, 16 problem proof pairs with 24 expert-annotated errors, built from data released in the First Proof project's second batch \cite{firstproof2}. \footnote{\texttt{ArxivMathGradingBench} is available at \url{https://huggingface.co/datasets/LukeBailey181Pub/ArxivMathGradingBench}. \texttt{FirstProofGradingBench} is available at \url{https://huggingface.co/datasets/LukeBailey181Pub/FirstProofGradingBench}}
\end{enumerate}

\section{The Pseudo-Formal Framework}

We first define Pseudo-Formal proofs, and then
state, with theoretical justification, desirable
properties of Pseudo-Formal proofs.

\begin{definition}[Pseudo-Formal proof]
\label{def:pf-proof}
Let a proof module be a tuple $(P, c, \pi)$ where $P$ is the set of premises, $c$ is the conclusion, and $\pi$ is the proof, which are all represented as natural language, that is, a sequence of characters. 

A \emph{Pseudo-Formal proof} is a sequence of modules $V=\{v_i = (P_i, c_i, \pi_i): i\in [n]\}$, equipped with a directed acyclic graph (DAG) $G$ for the invocation dependency and a forest $T$ for scope inheritance, whose node sets are both $V$. Here $G$  contains an edge
from $v_i$ to $v_j$ exactly when $\pi_i$ invokes the conclusion $c_j$, and 
the scope-inheritance forest $T$ contains 
$(v_i, v_j)$ exactly when $v_i$ inherits the scope (premises) of $v_j$ (as typically specified in the premises of $v_i$). For clarity, if $\operatorname{Anc}_T(v_i)$
denotes the set of ancestors of $v_i$ in $T$, then the premises available to
$\pi_i$ are $P_i \cup \bigcup_{v_j \in \operatorname{Anc}_T(v_i)} P_j$.

\end{definition}

Our proof format is intended to be both unambiguous and
modular: modularity comes from decomposing the proof into parts, while unambiguity comes from requiring each node to explicitly state
its premises, conclusion, proof, and dependencies.
The two associated graphs capture complementary structure in
a well-written human proof. The dependency graph $G$ 
records which proof modules invoke the conclusions of which others. We insist that it needs to be a DAG because this is a minimal structural requirement for the correctness of a proof (so that there is no circular argument).  The scope-inheritance forest $T$ hierarchically organizes the proofs into propositions, lemmas, and
finer-grained claims, and allows the sharing of context within each
scope. For example, after a statement such as ``In this section, let $S$ be
a vector space,'' subsequent claims can use $S$ without restating it. We insist that the scope-inheritance graph $T$ is a forest so that no proof module can inherit from two parents which may result in conflicting scopes.

The two graphs $G$ and $T$ are generally distinct but closely related. For instance, the proof of a proposition typically depends on the lemmas
introduced within its scope, which appear as its children in $T$. We note that our dependency graph is similar in spirit to a Lean blueprint~\citep{zhu2026leanarchitectautomatingblueprintgeneration} used in autoformalization.

The structure in the Pseudo-Formal proof induces a simple \textbf{block-wise verification algorithm}. We can verify each
node of $G$ in reverse order of its topological sort. In this way,
when verifying a node, we have already verified that all of its children
are correct, and thus have that any child the node invokes is correct.
To verify a node $(P_i, c_i, \pi_i)$ of the Pseudo-Formal proof we
simply (1) realize the full premises $P_{\text{full}}$ of the
node using the inheritance forest $T$, and (2) check that
$\pi_i$ serves as a valid proof of $c_i$ from $P_\text{full}$,
specifically checking that whenever $\pi_i$ invokes the conclusion
of a child node, it has proven the premises of said child hold
before doing so. 

However, a Pseudo-Formal proof may not always be modular and unambiguous: one could place an
entire complex proof inside a single node. We identify two desiderata that make a Pseudo-Formal proof easier to verify. 
First, \emph{conciseness}: each node, including its premises, conclusion,
and proof, should be short and invoke only a small number of other nodes in
$G$. Second, \emph{bounded scope depth}: the scope-inheritance forest $T$
should be shallow, so that the definitions and assumptions available to any
node can be recovered from only a short chain of ancestors. Together, these
conditions ensure that verifying each node requires only a \emph{small
amount of context}, rather than the context needed to verify the entire
proof at once. We call Pseudo-Formal proofs that adhere to these two desiderata \emph{Good Pseudo-Formal Proofs}, and provide formal justification that all 
\emph{Good Pseudo-Formal Proofs} are in principle easy to verify in~\Cref{thm:context-informal}.

\begin{theorem}[Context cost of block verification]
\label{thm:context-informal}
Fix a block-wise
verification algorithm $A$ (\Cref{def:block-wise-verifier}).
There exists a fixed-size Transformer $T_1$ such that for any good Pseudo-Formal proof $P$ (\Cref{def:good-pf-proof}),
$A$ on $P$ can be simulated by aggregating the results of $O(n)$ calls to $T_1$, with context length independent of $n$.
\end{theorem}

\Cref{thm:context-informal} states that block-wise verification
of a Good Pseudo-Formal proof requires context length independent of the size of the proof. This contrasts with line-by-line whole-proof verification by a
single fixed-size Transformer call: the proof input itself requires $O(n)$ context length, and the best construction requires a chain-of-thought of $O(n)$ length~\citep{merrill2023expressive}. By moving
from one call whose input grows linearly with proof length to many $o(n)$ token calls, modularization can place every verification step in the short-context regime.

The context length bound presented here should be viewed as a theoretical tool (adopted from~\citep{yang2025pencillongthoughtsshort,yang2026recursive}) to model context rot, the phenomenon
in which LLM performance on a task
decreases with input context size \citep{hong2025context,zeng2026loca}.
The theorem states that when verifying
a Good Pseudo-Formal Proof, the context
required for each individual verification
call is much smaller than that required
to verify a standard natural language proof.
Thus, due to context rot,
we expect the verification performance
on each node to be higher than a single large
verification call on an entire proof.

\begin{figure}[t]
    \centering
    \resizebox{0.86\linewidth}{!}{%
        \begin{tikzpicture}[
  font=\small,
  block/.style={
    draw=black, fill=white, rounded corners=1pt,
    align=left, inner sep=2pt, text width=112mm,
    minimum height=5.5mm
  },
  root/.style={
    block, draw=blue!70!black, fill=blue!6, line width=0.8pt,
    font=\bfseries\small
  },
  collapsed/.style={
    block, draw=black!65, fill=gray!8, text width=108mm,
    font=\small
  },
  prop/.style={
    block, draw=blue!55!black, fill=blue!3, text width=108mm,
    font=\small
  },
  lemma/.style={
    block, draw=black!75, fill=white, text width=103mm,
    font=\scriptsize
  },
  expanded/.style={
    lemma, text width=100mm, minimum height=9mm
  },
  bad/.style={
    expanded, draw=red!75!black, fill=red!8, line width=1pt
  },
  verifier/.style={
    draw=red!75!black, fill=red!4, rounded corners=2pt,
    align=left, text width=100mm, inner sep=3pt, font=\scriptsize
  },
  edge/.style={draw=black!75, line width=0.55pt},
  warn/.style={-{Stealth[length=2mm]}, draw=red!75!black, dashed, line width=0.8pt}
]
\node[
  root,
  anchor=north west
] (header) at (0,0) {
  $\blacktriangledown$ Thm: For $P$ satisfying a functional equation (omitted),
  $S=\{P(a)+P(-a):a\in\mathbb Q\}$, prove $S$ is finite and $\max |S|=2$.
};

\node[collapsed, anchor=north west] (p1) at ($(header.south west)+(7mm,-2mm)$) {
  $\blacktriangleright$ Prop. 1: Basic properties of $P$ \hfill {\scriptsize block score $7$}
};

\node[collapsed, anchor=north west] (p2) at ($(p1.south west)+(0mm,-1.5mm)$) {
  $\blacktriangleright$ Prop. 2: $L(a,b)\in\{0,-s(b)\}$ \hfill {\scriptsize block score $7$}
};

\node[collapsed, anchor=north west] (p3) at ($(p2.south west)+(0mm,-1.5mm)$) {
  $\blacktriangleright$ Prop. 3: Cauchy difference constraints for $D(x,y)$ \hfill {\scriptsize block score $7$}
};

\node[prop, anchor=north west] (p4) at ($(p3.south west)+(0mm,-2.5mm)$) {
  $\blacktriangledown$ Prop. 4: Structural properties of $S$ \hfill {\scriptsize block score $7$}
};
\node[expanded, anchor=north west] (l42) at ($(p4.south west)+(7mm,-1.5mm)$) {
  \textbf{Lem. 4.2 / P2.}\\
  \textbf{Statement.} If $d\in S$ and $d\ne0$, then $-d\notin S$. 
};

\node[prop, anchor=north west] (p5) at ($(l42.south west)+(-7mm,-2.5mm)$) {
  $\blacktriangledown$ Prop. 5: Upper bound, $S$ is finite and $|S|\le2$
};
\node[bad, anchor=north west] (l51) at ($(p5.south west)+(7mm,-1.5mm)$) {
  \textbf{Lem. 5.1.}\\
  \textbf{Assumptions / Conditions / Definitions.}
  Same as in Prop. 5 and Prop. 4 available.\\
  \textbf{Statement.}
  If $S^*=S\setminus\{0\}$ then $|S^*|\le1$
  (i.e. there is at most one nonzero element of $S$).\\
  \textbf{Proof.}
  By P2 of Prop. 4, all nonzero elements of $S$ must have the same sign;
  WLOG every element of $S^*$ is positive. \ldots
};

\node[verifier, anchor=north west] (resp) at ($(l51.south west)+(7mm,-1.5mm)$) {
  \textbf{Block verifier response.}\\
  The proof begins with an incorrect/unjustified deduction:
  ``By P2 of Prop. 4, all nonzero elements of $S$ must have the same sign;
  without loss of generality assume every element of $S^*$ is positive.''
  This does not follow from P2. Prop. 4 P2 only states that if $d\in S$
  and $d\ne0$, then $-d\notin S$, i.e. $S$ cannot contain a nonzero element
  together with its exact negative. It does not forbid the coexistence of
  distinct nonzero elements of opposite signs (for example, P2 allows both
  $2$ and $-3$ to be in $S$, since $-3$ is not equal to $-2$). Thus the
  claim that all nonzero elements have the same sign is unjustified.
};
\node[
  draw=red!75!black, fill=red!75!black, circle,
  inner sep=1.7pt, text=white, font=\bfseries\small,
  anchor=center
] at ($(resp.north east)+(-2.2mm,-2.2mm)$) {0};

\node[collapsed, anchor=north west] (p6) at ($(resp.south west)+(-14mm,-2.5mm)$) {
  $\blacktriangleright$ Prop. 6: Construction, $P(x)=2\lfloor x\rfloor-x$ gives $S=\{0,-2\}$
  \hfill {\scriptsize}
};

\draw[edge] ($(header.south west)+(2mm,0)$) |- ($(p1.west)+(-2mm,0)$);
\draw[edge] ($(header.south west)+(2mm,0)$) |- ($(p2.west)+(-2mm,0)$);
\draw[edge] ($(header.south west)+(2mm,0)$) |- ($(p3.west)+(-2mm,0)$);
\draw[edge] ($(header.south west)+(2mm,0)$) |- ($(p4.west)+(-2mm,0)$);
\draw[edge] ($(p4.south west)+(2mm,0)$) |- ($(l42.west)+(-2mm,0)$);
\draw[edge] ($(header.south west)+(2mm,0)$) |- ($(p5.west)+(-2mm,0)$);
\draw[edge] ($(p5.south west)+(2mm,0)$) |- ($(l51.west)+(-2mm,0)$);
\draw[edge] ($(header.south west)+(2mm,0)$) |- ($(p6.west)+(-2mm,0)$);

\draw[warn] (l42.east) to[bend left=16]
  node[midway, right, font=\scriptsize, align=left] {misread as\\``same sign''}
  (l51.east);

\end{tikzpicture}
    }
    \caption{
        \textbf{An example Pseudo-Formalized Proof on IMO 2024 P6}. 
        The block verifier correctly assigns Lemma~5.1
        score $0$ because the proof incorrectly strengthens
        P2 from Lemma~4.2
        ``$S$ cannot contain both $d$ and $-d$'' to ``all nonzero
        elements of $S$ have the same sign''.
    }
    \label{fig:imo-pf-tree-gb0813}
\end{figure}

\section{Pseudo-Formal Verification}

We now outline the process by which we take a natural-language proof $\Pi$
and use Pseudo-Formalization to verify its correctness.
We do so in four steps. (1)~An LLM
translates $\Pi$ into a Pseudo-Formal proof $\hat{\Pi}$, with an
optional self-repair loop that compares $\hat{\Pi}$ against $\Pi$ and
patches discrepancies. (2)~A block-verification LLM checks each
module of $\hat{\Pi}$ independently, producing a per-module error
report. (3) A calibration LLM aggregates these per-module errors
into an accept/reject decision. (4) All prior steps are 
run in parallel $k$ times and
the resulting verdicts are combined pessimistically (a proof 
only passes if all rollouts accept it).
Each step
calls a single off-the-shelf
LLM with no fine-tuning. The remaining subsections describe each
step in detail.

\subsection{Step 1: Translation to Pseudo-Formal}

The translation LLM receives the
natural-language proof $\Pi$, and is prompted to produce a Pseudo-Formal
proof $\hat{\Pi}$ in the format of
Definition~\ref{def:pf-proof}.
The translation process can create a Pseudo-Formalization
that is not faithful to the original proof. To mitigate this, we run
an optional
\emph{self-repair} loop. For $k$ iterations, an LLM is shown
$\Pi$ together with the current $\hat{\Pi}$ and is asked to identify
discrepancies and produce a patched version, or, if there are no
discrepancies, exit the loop. Self-repair is not fundamental to our
method, but it is a simple and effective way to improve the quality of the translation.

\subsection{Steps 2 \& 3: Block verification and calibration}

\paragraph{Block verification.}
Given $\hat{\Pi}$, the block-verification LLM is invoked once per
node $v = (P, c, \pi)$. The LLM is shown the node's
full premises ($P$ combined with all ancestor premises
derived from $T$), conclusion, and proof, the
premises and conclusions of any child nodes,
and is asked whether $\pi$
correctly establishes $c$, assuming all child nodes are correct.
It returns either an accept
verdict or a description of the error (see \Cref{fig:imo-pf-tree-gb0813}).

\paragraph{Calibration.}
Aggregating per-module error reports into a single verdict is
non-trivial. The block verifier may flag genuine reasoning errors,
but it may also raise minor objections, such as typos or
harmless ambiguities. In some cases we may want to catch such
errors (before publishing a paper, for example) but in others
we may not. We delegate the aggregation of
errors to a
\emph{calibrator} LLM, which receives the full list of per-module
error reports together with a natural language description
of the desired strictness threshold, and emits a final
verdict. This verdict could be a binary correctness label,
a numerical score, or something more expressive (e.g. in \Cref{sec:arxiv}
we have the calibrator output specific locations of errors within
the original natural language proof).
The calibration step allows
a downstream user to tailor the verification process to their needs.

\subsection{Step 4: Parallel scaling}

Steps 1--3 are stochastic and can be run independently $k$ times, with the
resulting verdicts combined into a final output. Following
\citet{huang2025pessimistic}, who demonstrate that pessimistic aggregation
is effective for math verification, we accept a proof only if every one of
the $k$ runs accepts it (a single flagged error suffices to reject).
Framing proof verification
with ``contains error'' as the positive class, increasing
$k$ can only add flagged errors and never remove them, so false positives
increase and false negatives decrease as $k$ grows. Consequently,
recall tends to rise with $k$ while precision tends to fall. Scaling $k$
thus exposes a family of operating points that trade recall against
precision, and in our experiments we sweep $k$ and plot the resulting
curves to characterize this tradeoff.
Different downstream settings prioritize precision or recall to varying
degrees, with pessimistic parallel scaling offering a method for
navigating the tradeoff.

\section{Experimental Results}

We evaluate Pseudo-Formal Verification on two
benchmarks of two different mathematical difficulties and proof 
styles: competition math (at the level of International Math 
Olympiad and Putnam) and research level proofs
drawn from arXiv and First Proof \cite{firstproof2}. 
For competition math, we
take the existing benchmark \texttt{Hard2Verify}. For
research level, we create new datasets: \texttt{ArxivMathGradingBench} based on in-the-wild arXiv revision comments and \texttt{FristProof\\GradingBench} based on the First Proof second 
batch data.

\begin{figure}[t]
    \centering
    \includegraphics[width=\linewidth]{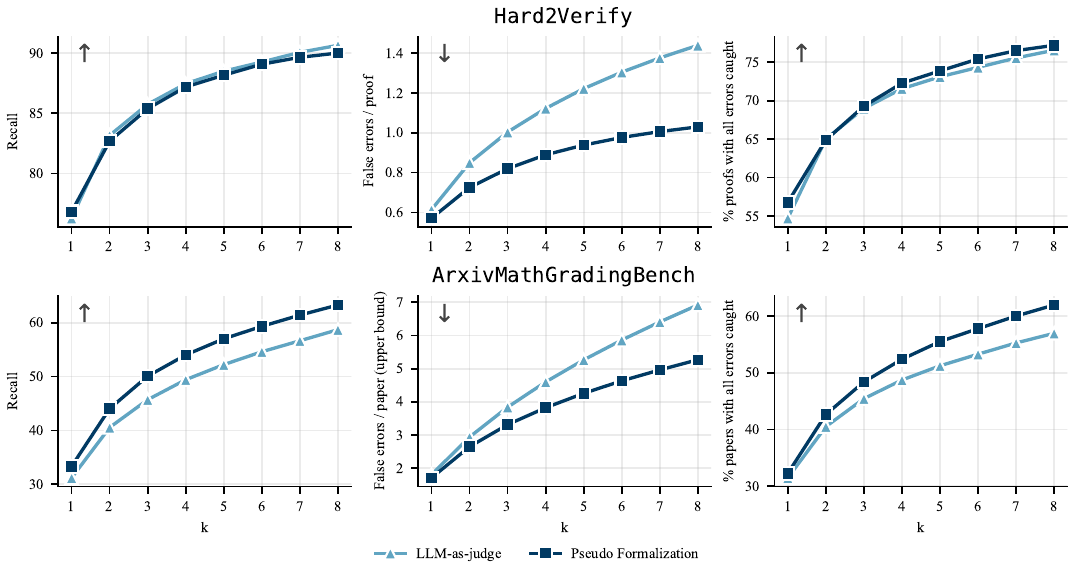}
\caption{Per-proof/paper and per-error metrics on \texttt{Hard2Verify} and \texttt{ArxivMathGradingBench}. Here the x axis is always the number of parallel verification attempts $k$. 
    \textbf{Left: Recall (per step for \texttt{Hard2Verify}, per error for \texttt{ArxivMathGradingBench}).} We perform similarly to the baseline on \texttt{Hard2Verify}, both recovering around $90\%$ of the errors and significantly outperforming baseline on the harder and longer dataset of \texttt{ArxivMathGradingBench}. 
    \textbf{Middle: Mean number of false errors per proof/paper (lower is better).} We define false errors as predictions
    that do not match any annotated ground-truth error. This is  an \emph{upper
  bound} on the false-alarm rate on \texttt{ArxivMathGradingBench}, as we only annotate
  errors disclosed by the authors. We outperform baseline on both datasets.
  \textbf{Right: Coverage of true errors per proof/paper.} We define coverage as the portion of proofs/papers in which the verifier covers \emph{every} ground-truth error, in which we match baseline for \texttt{Hard2Verify} and exceed 
  it on \texttt{ArxivMathGradingBench}.}
    \label{fig:arxiv_results}
\end{figure}

\subsection{Competition Math}
\label{sec:competition-math}

We test our method on \texttt{Hard2Verify}~\citep{pandit2025hard2verify} an existing proof-verification benchmark at the competition math level.
It consists
of 200 proofs generated by frontier AI models on 80 
olympiad-level and Putnam problems, paired with step-level
binary correctness labels for each proof. An entire proof is
labeled as correct if all steps are labeled correct, and incorrect 
otherwise.

\textbf{Metrics.} Following~\citet{pandit2025hard2verify}, a verification system
being evaluated on \texttt{Hard2Verify}
accepts a problem and a proposed solution split into indexed steps,
and returns a binary correctness prediction for each step.
Since the benchmark annotates every step, we evaluate the verifier's
ability to identify erroneous steps using precision and recall.
For each proof $x$ in dataset $\mathcal{D}$, let
$Y^x$ denote the set of ground-truth
incorrect step indices, and let
$\hat{Y}^x$ denote the set of step
indices predicted to be incorrect by the verifier. We define
\begin{align}
\label{eq:metric}
\text{TP}^{\text{step}} = \sum_{x \in \mathcal{D}}|\hat{Y}^x\cap Y^x|, \quad
\text{FP}^{\text{step}} = \sum_{x \in \mathcal{D}} 
|\hat{Y}^x \setminus Y^x|,\quad
\text{FN}^{\text{step}} = \sum_{x \in \mathcal{D}}  | Y^x \setminus \hat{Y}^x|
\end{align}

We report the step-wise precision and recall as
$\text{Precision}^{\text{step}} = \text{TP}^{\text{step}}/(\text{TP}^{\text{step}}+\text{FP}^{\text{step}})$ and
$\text{Recall}^{\text{step}} = \text{TP}^{\text{step}}/(\text{TP}^{\text{step}}+\text{FN}^{\text{step}})$.

We additionally report proof-level metrics by aggregating step-level
predictions. A proof is labeled
incorrect if $Y^x \neq \emptyset$, and correct otherwise; similarly,
the verifier predicts proof $x$ to be incorrect if
$\hat{Y}^x \neq \emptyset$. Applying this mapping to every proof gives
proof-level TP, FP and FN counts, from which we compute
proof-level precision and recall with the same formulas as above. We will denote these terms as $\text{Precision}^{\text{proof}}$ and $\text{Recall}^{\text{proof}}$. 
Ideally, a verification system should be able to find all the right errors in one proof without flagging any fake errors. This is already reflected partially in the step level precision and recall but we further perform a fine-grained proof level breakdown. We track the coverage metric per proof $\text{Coverage}^{\text{proof}} = \mathrm{E}_{x \sim \mathcal{D}}[ \mathbf{1}(Y^{x} \subseteq \hat Y^{x}) \mid Y^x \neq \emptyset]$ and the number of false errors per proof $\text{\#FE}^{\text{proof}} =  \mathrm{E}_{x \sim \mathcal{D}} [|\hat Y^x \setminus  Y^x|]$.

\textbf{Methods.} For baseline, we use direct LLM-as-judge 
using the same prompt as in the \texttt{Hard2Verify} paper.
The model is prompted with the problem and solution broken
up into steps, and asked to identify which steps have errors. 
This provides step-wise error predictions. If any step is 
predicted to be incorrect, then the proof as a whole is 
predicted to be wrong.
PF+BV instead rewrites the proof into the 
format of \Cref{def:pf-proof}, verifies each module, and then 
calibrates the block-level reports back into the benchmark's output format.
We do not provide the step-wise decomposition to the translation LM 
that produces the PF proof. We do however provide it to the calibration
LM, and ask the calibrator to identify from the block errors, which 
of the original steps were correct and incorrect. This way, PF+BV 
provides the same exact type of output as the baseline.
We perform parallel scaling by pessimistically aggregating predictions over each step across
$k=8$ independent runs. This means that we predict a step as incorrect 
if it was predicted as incorrect in any of the parallel runs, for 
PF+BV and the baseline.
We use GPT-5.4 mini for all the results here (cf. more details in \Cref{app:hard2verify-prompts}, we 
also test more models in
\Cref{fig:pr_diff_models}).

\begin{figure}[t]
    \centering
    \includegraphics[width=\linewidth]{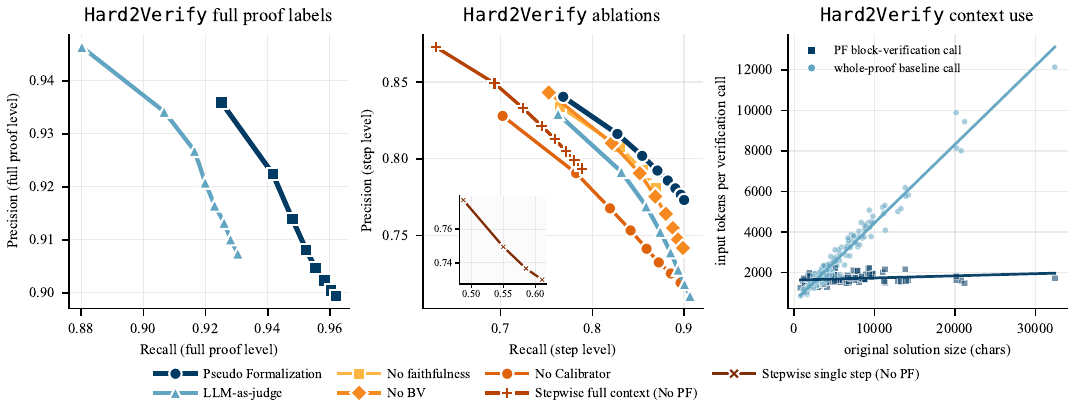}
\caption{ \textbf{Left:}
Precision-recall on \texttt{Hard2Verify} at the
\textbf{whole-proof label level} (positive class: proof contains
1 or more incorrect steps). For each $k$, the 
verifier's per-rollout proof-level verdicts are aggregated with
pessimistic-min (the proof is predicted incorrect iff at least one
1 or more steps over the $k$ rollouts is predicted incorrect). 
\textbf{Middle:}
Ablations on Method. We report step-level precision and recall for each variant. \textit{No faithfulness} drops the PF faithfulness check. \textit{No BV} removes block verification and directly grades the translated PF proof in one call. \textit{No Calibrator} maps BV block flags to steps without independent re-grading. \textit{Stepwise full context} removes the PF and grades each original step separately with the full solution as context. \textit{Stepwise single step} instead grades each step in isolation.
\textbf{Right:} Average length of context used for 
whole proof verification and block verification call 
against input proof size. We see that 
the block sizes of PF do not grow with original 
proof length, adhering to \Cref{def:good-pf-proof} 
in \Cref{thm:context-informal}.
}
\label{fig:pr_full_proof}
\end{figure}

\textbf{Quantitative Results.} \Cref{fig:figure_1} top right shows the 
effect on $\text{Precision}^{\text{step}}$ and $\text{Recall}^{\text{step}}$ when we scale the number of
independent runs on \texttt{Hard2Verify}.  
Our algorithm achieves a better Pareto 
precision--recall tradeoff compared to the baseline. 
Top left of \Cref{fig:arxiv_results} shows the recall with 
respect to $k$ where we perform comparable to the baseline. 

\Cref{fig:pr_full_proof} left shows the full proof precision 
recall curve (where a proof is labeled wrong if any step 
is wrong, and predicted wrong if any step over $k$ rollouts 
is predicted wrong). We see our method outperforms in terms 
of recall by a large margin, but does not dominate in terms 
of precision. However, 
top middle of \Cref{fig:arxiv_results} shows that our $\#\mathrm{FE}^{\mathrm{proof}}$ is actually lower. We flag $28\%$
fewer false steps per proof compared to baseline at $k=8$. Top right of \Cref{fig:arxiv_results} shows that our coverage, the portion of proofs where our verifier flags \emph{every} ground truth error, is also higher than baseline. 

\Cref{fig:pr_full_proof} middle shows a number of 
ablations, with ablation 
descriptions given in the 
figure caption. Each performs below the Pareto 
frontier of our full method.  Interestingly, removing faithfulness check or BV both outperforms baseline, showing the usefulness of the translated PF proof. \Cref{fig:pr_full_proof} right 
shows the average length of block verification 
calls against input proof size for \texttt{Hard2Verify}.
We see that the context size for the block verification 
calls \emph{does not} grow with the input proof size,
showing empirically the Pseudo-Formal proofs we produce 
adhere to \Cref{def:good-pf-proof} used in \Cref{thm:context-informal} 
(noting that tree depth is capped at 3 for hard to verify, as we prompt 
for theorem, proposition, lemma hierarchy).

\textbf{Qualitative Analysis.} \Cref{fig:imo-pf-tree-gb0813} illustrates the mechanism on a
representative successful run. In this solution, an earlier proposition
shows only that $S$ cannot contain both $d$ and $-d$. The later proof of
Lemma~5.1 uses that statement as though it implied the stronger claim that
all nonzero elements of $S$ have the same sign. The Pseudo-Formal rewrite
isolates this jump into a local block: the block verifier can then compare the proposition isolated for it in context against the lemma proof and flag the strengthening
without re-verifying the whole solution at once.

\subsection{Research Papers}
\label{sec:arxiv}

To stress-test verification on full research-paper proofs, we
construct a benchmark from arXiv math papers
that contained errors later corrected by their authors.
We filter math papers posted to arXiv between 2023 
and 2025 inclusive and identify
those that were updated with comments that explicitly mention fixing
a flawed lemma, theorem, or proposition. We do so using first a regex and
then GPT 5.4 nano 
(see Appendix \ref{app:arxiv-paper-filtering}
for more details).
For each paper, we use the PDF as 
the verifier input.
We manually checked 
each entry proposed 
by this procedure until
we had 100 
examples of papers with a known specific error
location, which we refer to as
\texttt{ArxivMathGradingBench}. Each entry
of this dataset is a tuple containing (1) the paper PDF, and (2) the
location of error(s) in the paper (e.g. ``Lemma 1.1'' or ``Theorem A, Proposition A.''). We provide examples of three dataset
entries in \Cref{tab:arxiv-examples}.

\texttt{ArxivMathGradingBench} allows us to test two
important properties of our verification method.
Firstly, as the dataset is made from real
research-level papers, success on it directly
transfers to real-world applications (aiding
the referee process of mathematics papers, and
as a reward signal for reinforcement learning on frontier
models). Secondly, like \texttt{Hard2Verify}, 
each entry
has a specific known error location. This allows
us to test the ability of verifiers to
find the specific error in a long research level
proof.

\begin{table}[t]
\centering
\footnotesize
\caption{Example rows from the \texttt{ArxivMathGradingBench} dataset.
Each row is a published arXiv math paper whose author later flagged and
corrected an error in a subsequent version. The \emph{Comment} column
is the author's revision note verbatim, and \emph{Error Location}
gives the error location extracted from the comment.}
\vspace{4mm}
\label{tab:arxiv-examples}
\begin{tabularx}{\linewidth}{@{}p{3.3cm} p{2.7cm} X p{1.8cm}@{}}
\toprule
\textbf{Title} & \textbf{Category} & \textbf{Comment} & \textbf{Error Location} \\
\midrule
Power-free palindromes and reversed primes \citep{chourasiya2025power}
& math.NT \newline (Number Theory)
& 14 pages, 1 figure, Minor corrections to Theorem 1.5 and Lemma 2.1.
& Theorem 1.5, Lemma 2.1 \\
\addlinespace
Real algebraic surfaces biholomorphically equivalent but not algebraically equivalent \citep{rond2025real}
& math.CV \newline (Complex Variables)
& The proof of Lemma 5 is incorrect: we do not have $H(M_{\mathcal S}) \subset M_\infty$ since $M_\infty$ is defined by 2 equations and $\mathrm{Im}(w)=0$ has not been considered in the proof. This affects the main result, and we do not know whether it is true.
& Lemma 5 \\
\addlinespace
Finite codimension stability of invariant surfaces \citep{forni2025finite}
& math.DS \newline (Dynamical Systems)
& Lemma 4.1 has been corrected after N. Tedesco pointed out a mistake in the calculations. The main results are unchanged.
& Lemma 4.1 \\
\bottomrule
\end{tabularx}
\end{table}

\textbf{Metrics.} A verification system
being evaluated on \texttt{ArxivMathGradingBench}
accepts the PDF  of a paper,
and then returns a list of error locations.
For each returned error location, we use
an LLM judge (GPT 5.4 mini) to decide
if this matches
any of the ground-truth error location(s) (see Appendix \Cref{sec:app_arxiv_judge} for details). 
This is different from the step-wise prediction in \texttt{Hard2Verify} as there is no natural concept of step here. 
In this setting, where every dataset entry
has a variable number of error locations
that the verifier must predict, we consider
the analogue of precision and recall as
our primary evaluation metrics.
For each problem $x$ in dataset $\mathcal{D}$,
with ground-truth error locations $ Y^x = \{y_i^x\}_{i=1}^{k_x}$
and verifier predicted errors $ \hat{Y}^x = \{\hat{y}_i^x\}_{i=1}^{m_x}$. We can then define
metrics including the step-wise precision and recall in~\Cref{eq:metric} in the same way as the previous section.

Importantly, because $Y^x$ is a \emph{subset} 
of possible errors in the proof (the author may not 
know about other errors, or may have corrected them but 
not annotated this in the arXiv comment), our False Positive values
are \emph{upper bounds} on the true number of False
Positives (the verifier may be identifying a valid error
that is simply not reflected in $Y_x$) meaning 
that our precision values are \emph{lower bounds} on the 
true values.

\textbf{Methods.} 
As before, we use GPT-5.4 mini for the baseline
and all stages of our method. 
We modify the 
calibrator and baseline prompts to output a list of 
errors and their locations in the original 
PDF instead of binary classification per step (see Appendix \ref{app:arxiv-prompts} for more details).
For parallel 
scaling, we run $k$ independent verification 
runs, and take the union of all errors (by the 
error location) found 
in each run (thus deduplicating error locations 
found in multiple parallel runs).

\textbf{Evaluation.}
\Cref{fig:arxiv_results,fig:figure_1} shows
verifier performance on
\texttt{ArxivMathGradingBench}. Firstly 
we note the benchmark is very challenging, as we would 
expect. Even with $k=8$ rollouts both verifiers leave
roughly $40\%$ of annotated ground-truth errors uncaught and predict
several errors per paper that do not match any annotation
(low precision in Figure \ref{fig:figure_1}). Some of
this apparent 
\begin{wrapfigure}{l}{0.5\textwidth}
    \centering
    \includegraphics[width=0.48\textwidth]{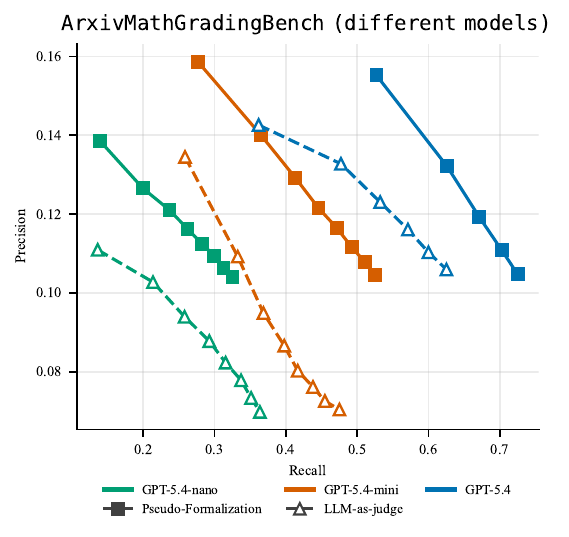}
    \caption{Performance of
    different models on a 35 problem
    subset of \texttt{ArxivMathGradingBench}. 
    The benefit of PF persists as the 
    model improves. Each point 
    from top left to bottom right 
    indicates increasing 
    $k$ by 1.}
    \label{fig:pr_diff_models}
    \vspace{-7mm}
\end{wrapfigure}
overflagging is unavoidable, since
labels are for specific mistakes that the
authors disclosed in revision comments. Any genuine-but-unannotated
issues the verifier raises are also counted as non-matching, making
the middle panel of \Cref{fig:arxiv_results} an \emph{upper bound} on
the true false positive rate, and thus precision in \Cref{fig:figure_1}
a lower bound.
Even so, Pseudo-Formalization compares favorably to the
LLM-as-judge baseline on every axis we measure.
In \Cref{fig:figure_1}, we see a Pareto-dominant precision--recall lower-bound curve. In \Cref{fig:arxiv_results},
at every value of
$k$ the Pseudo-Formal verifier compared favorably with the baseline. It 
reaches higher per-error recall, raises fewer false alarms
(middle), and catches every annotated error in a larger
fraction of papers (right).

In \Cref{fig:pr_diff_models} we 
show how the performance 
of PF with different 
base models on a 35 problem subset
of \texttt{ArxivMathGradingBench} 
(we subset the data due to resource constraints). 
For GPT 5.4 full to 
reduce cost of running PF we 
do not include any self-repair 
on the translation step.
We see 
across all of GPT 5.4 nano, mini, 
and full that PF outperforms the 
baseline LLM-as-judge considerably.
This demonstrates (a)
the benefits of PF 
are not isolated to GPT 5.4 mini, 
the main model we use throughout 
the paper, and (b)
the benefit
of PF \emph{composes} with 
model scaling, 
naturally desirable
property to have for the method.

\begin{table}[t]
\centering
\footnotesize
\caption{Example errors from the \texttt{FirstProofGradingBench} dataset.
Each dataset entry is a (problem statement, submitted proof, error list)
triple, where every error is a (span, description) pair and the span is a
verbatim excerpt occurring exactly in the submission. The submitted proofs are not shown here. Unlike
\texttt{ArxivMathGradingBench}, the ground truth states \emph{what} is wrong,
not only where.}
\vspace{4mm}
\label{tab:firstproof}
\begin{tabularx}{\linewidth}{@{}p{4.0cm} X X@{}}
\toprule
\textbf{Problem} & \textbf{Error span (verbatim from submission)} & \textbf{Error description} \\
\midrule
Let $\Gamma$ be a finite simplicial graph with $|V(\Gamma)|\geq 3$, and
$(M_v,\varphi_v)$ be a separable tracial von Neumann algebra for each $v\in V$.
If $[\dots]$
then the graph product von Neumann algebra $M$ is properly proximal.
& A DKEP boundary piece $X\subset B(H)$ means a hereditary $C^*$-subalgebra,
invariant under the left and right $M$-actions, whose elements are treated as
$X$-compact operators.
& This sentence gives an incorrect and imprecise definition of the boundary
pieces used in the DKEP theory. $[\dots]$ \\
\addlinespace
Let $M$ be a matroid, and let $\operatorname{Trop}(M)$ be the corresponding
tropical linear space. $[\dots]$ Is there an order-reversing involution $\Phi$ on
$\operatorname{Dr}(M)$ that extends the involution $F\mapsto F^{\perp}$?
& Finally, the valuated flag relation $q\twoheadrightarrow p$ is carried to
$p^{\perp}\twoheadrightarrow q^{\perp}$, which is precisely the asserted
reversal of containment of tropical linear spaces.
& The proof asserts, without calculation or argument, that the tropical
incidence relations defining a valuated quotient are transformed by the lattice
polarity into the reversed quotient relation. $[\dots]$ \\
\bottomrule
\end{tabularx}
\end{table}

\section{Expert Refereed Research Math}

One shortcoming of \texttt{ArxivMathGradingBench} 
is that we only have coarse-grained 
labels of error locations, as opposed 
to details of what the errors are. 
We introduce a new benchmark to 
resolve this issue, 
\texttt{FirstProofGradingBench}. 
We build this benchmark from the First Proof project \cite{firstproof2}, an
ongoing effort to measure the ability of AI systems to solve research-level
mathematics. There have been two rounds of problems released. In both 
rounds, the problems are contributed by mathematicians who encountered 
said problems during their own research, and no frontier model at the time of 
release could solve them. The second round, referred to as the ``second batch'', 
collected proof attempts from four different AI systems and 
had 2-3 mathematicians referee the proofs. The 
resulting referee reports recorded not only where an error occurred 
but also details of what the exact error was, precisely the fine-grained ground 
truth that
\texttt{ArxivMathGradingBench} lacks.

From the $39$ second-batch solutions we construct \texttt{FirstProofGradingBench} as
follows. We first (upon the advice of the First Proof authors) discard 
all solutions from one of the four systems whose
proofs contain a large number of errors each.
Of the remaining solutions we keep only those on which the referees
agree. Either every referee rated the solution as requiring major
revisions, or every referee rated it essentially flawless. 
This leaves
$16$ proofs, $13$ flawed and $3$ correct.

Because multiple referees report the same errors, we consolidate them 
to create deduplicated ground truth error labels. For each submission we give 
GPT 5.5 the solution together
with \emph{all} of its referee reports, and ask it to output 
for each unique error raised by referees, a verbatim excerpt 
of the submission containing the defect, a description of what 
is wrong with it, and the referees who raised
it. We check that each excerpt occurs exactly in the submission text, so every
error is associated with source text. This
yields $24$ ground-truth errors across the $16$ proofs (the $3$ correct
solutions contain none by definition). Each entry of the dataset is a tuple
containing (1) the original problem statement, (2) the submitted solution, and
(3) the list of ground-truth errors, each given as a (span, description) pair.
We provide examples in Table~\ref{tab:firstproof}.

\textbf{Metrics.} A verification system evaluated on
\texttt{FirstProofGradingBench} accepts a problem statement and a proposed
solution, and returns a list of errors, each quoting the passage it objects to
and describing the objection. Since predictions and ground-truth errors are
both free text, we use an LLM judge
(GPT 5.5) to adjudicate the match. The 
judge is shown
the solution, the ground-truth errors, and the predicted errors, and reports
for \emph{each prediction} which ground-truth errors it matches. 
We refer to this evaluation system
that uses an LLM  match
predicted errors to ground 
truth errors as \emph{verifier auto eval} (additional details 
in Appendix \Cref{sec:verifier-auto-eval}).
As in prior sections 
we report precision and recall,
computed as in \Cref{eq:metric}, where 
a true positive is counted when 
a ground-truth error is matched 
by any prediction 
and a prediction that 
matches no
ground-truth error counts as a 
false positive. 

\texttt{FirstProofGradingBench} improves on \texttt{ArxivMathGradingBench} in
two ways. (1) A verifier must have the 
correct error reason, as opposed 
to just its location, and
(2) the ground-truth error set is not a subset of the errors present.
False positives are therefore countable, and precision is measured
directly rather than upper-bounded. 
These stronger evaluation criteria are designed in part to compensate for the benchmark's smaller sample size.

\textbf{Methods.} \texttt{FirstProofGradingBench} is substantially harder 
than any of the prior benchmarks. We accordingly upgrade to 
GPT 5.5 high reasoning effort. This incurs a larger experimental cost, 
and so we modify our method to produce 1 PF per proof, and
then reuse this across parallel BV (whereas previously we generated a new PF for each of the $K$ runs). Producing a single PF 
greatly reduces the cost of running our method. 
For baseline and our block verification step, we extend the prompt 
to (1) check objects in the proof are well defined, including 
web searching for cited definitions, (2) check every cited 
result is used correctly, and (3) try to find counterexamples.
As we enable web search, we ban GitHub and the First Proof website, 
and audit the list of websites the agents visit to ensure 
they do not cheat by accessing the ground truth referee reports.

To both methods we also add an additional \emph{theorem check}
step before the main verification, that checks if the proof 
is indeed even tackling the assigned problem.
For baseline we prompt the model
with the original problem statement, and the proof, and ask if the 
statement in the proof matches the original problem. For our method 
we prompt with the PF's theorem block instead of the entire proof.
We provide additional 
details of the method 
in Appendix \Cref{app:firstproof-prompts}.

\textbf{Evaluation.}
  \Cref{fig:figure_1} shows precision and recall for both methods as we scale
  the number of parallel verification runs $K$. With a single run our method
  recovers $0.66$ of the referee-identified errors at $0.77$ precision,
  against
  $0.54$ recall at $0.76$ precision for the baseline. Unlike \texttt{Hard2Verify} and
  \texttt{ArxivMathGradingBench}, pessimistic aggregation buys little
  additional
  recall for the baseline. The baseline's is essentially flat in $K$, giving a near-vertical curve, while ours
  rises from $0.66$ to $0.71$.

\section{Related Works}

Prior work has studied the use of LLMs to verify
mathematical proofs. Some explore training 
verifiers using a combination of supervised fine-tuning and
reinforcement learning \citep{shao2025deepseekmath,shi2025heimdall,gao2025mathminos,chen2025spc,zhang2025generative,wang2025solvingverifying,liu2025trust}.
This direction is complementary to ours in that one can use 
training to improve the performance of a Pseudo-Formal verifier, 
although we leave this to future work.
Other work explores
how to integrate an LLM verifier into proof generation, both as a
training signal and as a guide for test-time search.
\citet{huang2025winning} build a model-agnostic
verification-and-refinement loop that, when paired with frontier
LLMs, solves five of six IMO 2025 problems.
\citet{mahdavi2025scaling} scale tournament-style pairwise comparison
between candidate proofs.
\citet{dang2026escapingcognitivewellefficient} notice the existence of a grader failure mode in the solver-grader loop called \emph{Cognitive Well}, where the grader fails to identify a basic mistake in a family of proof candidates. They addressed this issue through \emph{conjecture extraction}, isolating candidate lemmas from the generated proof and verifying them in a fresh context. 
\citet{zhao2025sample} identify a phenomenon they call
\emph{implicit scaling}, where as the candidate pool of proofs grows, 
self-verification accuracy rises in part because larger pools tend to contain proofs
that are more legible to the verifier.
Pseudo-Formalization aims at
the same effect more directly, by translating a single proof into a
form that is explicitly easier for the verifier to check. Closest to
our work is \citet{huang2025pessimistic}, who present 
a method of verifying fixed size chunks of $L$ lines 
of a proof in parallel. We share the motivation that verification 
focused on certain parts of the proof is beneficial (modularization),
but use Pseudo-Formalization to determine the chunks to verify.

\emph{Autoformalization} refers to the process of translating a
natural-language proof into a fully formal language such as Lean or
Isabelle using an LLM \citep{szegedy2020promising,wu2022autoformalization}. The
technique has been used to generate large-scale training data for
formal theorem-proving models \citep{lin2025goedel,lin2025goedelv2}
and to support ongoing efforts to formalize existing research-level
proofs \citep{kontorovichpnt}. Pseudo-Formalization shares
autoformalization's high-level recipe of using an LLM to translate a
proof into a more checkable format, but the goals differ.
Autoformalization aims at producing certainty that a proof is correct,
at the expense of tractability in many settings.
Pseudo-Formalization gives up the correctness
guarantee in exchange for being applicable to mathematics that
current formal libraries cannot express, and for letting the
verifier exploit the rich vocabulary and standard techniques that
mathematicians already use. We see the two approaches as
complementary. Where autoformalization is feasible, it provides
stronger guarantees; where it is not, Pseudo-Formalization offers an
alternative.

Pseudo-Formalization also connects to the literature on
\emph{scalable oversight}, which studies how weaker supervisors can
reliably evaluate the outputs of more capable models. A central idea
in this line of work is the prover-verifier game, in which a prover
is trained to produce outputs that a weaker verifier can check, with
the explicit objective of making the proofs legible
\citep{anil2021learning,kirchner2024prover}. Pseudo-Formal proofs
are a legible format by design. This makes Pseudo-Formalization a
natural target for training models to produce more legible outputs.

\vspace{-1.0em}
\section{Conclusion and Limitations}
\vspace{-0.5em}
In this work, we introduced Pseudo-Formalization (PF) and Block Verification (BV), a framework designed to automatically and reliably verify mathematical proofs. 
By bridging the gap between the flexibility of natural language and the modularity and unambiguity of fully formal languages, PF decomposes complex arguments into explicit, self-contained modules, and organizes them into semantically meaningful DAG structures. To facilitate robust evaluation in the latter setting, we introduced \texttt{ArxivMathGradingBench} and 
\texttt{FirstProofGradingBench}, challenging benchmarks comprising  research level proofs with known errors. We show theoretically and empirically that using PF+BV can be favorable in error identification compared to standard LLM-as-a-judge on both competition and research level math.

\subsection{Limitations and future work}
\label{sec:limitation}

\textbf{Translation to Pseudo-Formal}. \Cref{thm:context-informal} 
characterizes 
the difficulty of verifying a Good Pseudo-Formal proof, but takes
the existence of such a proof as given. In practice producing 
a faithful Pseudo-Formal rewrite is not trivial. While empirically we are still able 
to improve performance when translation to pseudo-formal 
is necessary, a natural direction is to train models to
produce proofs directly in Pseudo-Formal format. This connects to work on making model
outputs more legible to verifiers, including prover-verifier games
\citep{kirchner2024prover}, for which Pseudo-Formalization 
provides a legibility target. In this case, outputting in Pseudo-Formal format adds a small additional burden on the provers but this should be significantly easier than coming up with correct proofs.

\textbf{Application Area}. We evaluate PF+BV only on mathematical proofs. Mathematics is a natural starting point because its proofs are explicitly structured around premises, conclusions, and chains of justification. However, many domains outside mathematics (e.g. empirical sciences), share this structure, and we expect the core principles of modularization and explicit premise-conclusion decomposition to transfer to these settings, where a `pseudo-formal argument' would be easier to scrutinize than a natural language one. 

\textbf{Error labels}. For
\texttt{Hard2Verify} and \texttt{ArxivMathGradingBench}
we use error localization 
as a proxy for true error identification.
In \texttt{FirstProofGradingBench} 
we use more detailed error details, however 
the benchmark is very small due to the difficulty 
of acquiring this type of supervision.
Future work should attempt to scale these 
types of benchmarks where 
 human mathematicians label error location 
 and reasons.

\section{Acknowledgments}

LB thanks the support of a Stanford Graduate 
and Vitalik Buterin Fellowship. KW thanks the support of the Stanford Graduate Fellowship.
TM thanks the support of NSF 2522743. MA thanks the support of NSF 2506145.
 We thank Neil Band, Caroline Choi, Thomas Chen, and Arvind Mahankali for feedback on an early draft of this work. 
We thank Marka Ellertson, Joshua Bailey, and Alexander Bailey 
for help with manuscript writing.

\newpage
\bibliographystyle{plainnat}
\bibliography{references}

\newpage
\appendix

\section{Context cost of verification: formal statement and proof}
\label{app:context}

\begin{definition}[Good Pseudo-Formal proof]
\label{def:good-pf-proof}
A \emph{good Pseudo-Formal proof} is a Pseudo-Formal proof (\Cref{def:pf-proof}) with the following properties for some constants $D,L$ independent of $n$,
\begin{enumerate}
    \item The depth of $T$ is $D < n$.
    \item The lengths of the conclusion $c$, the premises, and the proof are bounded above by $L$.
\end{enumerate}
\end{definition}

\begin{definition}[Block-wise verification algorithm]
\label{def:block-wise-verifier}
Let $P=(G,T)$ be a Pseudo-Formal proof with node set
$V=\{v_1,\ldots,v_n\}$, where each node is written as
\[
    v=(P_v,c_v,\pi_v).
\]
For a node $v\in V$, let $\Anc_T(v)$ denote its ancestors in the
scope-inheritance forest $T$, and define its dependency set by
\[
    \Dep_G(v)
    :=
    \{u\in V:(v,u)\in E_G\},
\]
namely the set of modules whose conclusions are invoked by $\pi_v$.

The \emph{module context} of $v$ is
\[
    C_v(P)
    :=
    \Bigl(
        (P_u)_{u\in \Anc_T(v) \cup \Dep_G(v)},
        (c_u)_{u\in \Dep_G(v)},
        P_v,c_v,\pi_v
    \Bigr),
\]
where all ancestor and dependency nodes are listed according to the canonical
left-to-right depth-first order of the proof forest.

A \emph{block-wise verification algorithm} $A$ is an algorithm specified
by a fixed module verifier $V_{\mathrm{mod}}$. On input $P$, it runs
\[
    V_{\mathrm{mod}}(C_v(P))
\]
for every module $v\in V$, and accepts if and only if all module checks accept:
\[
    A(P)=\accept
    \quad\Longleftrightarrow\quad
    V_{\mathrm{mod}}(C_v(P))=\accept
    \text{ for every } v\in V.
\]
We assume that $V_{\mathrm{mod}}$ halts on every module context in time \textbf{linear}
in the serialized length of that context.
\end{definition}

\begin{theorem}[Formal version of~\Cref{thm:context-informal}]
\label{thm:context-formal}
 Fix a block-wise
verification algorithm $A$ (\Cref{def:block-wise-verifier}).
There exists a fixed-size Transformer $T_1$ such that for any good Pseudo-Formal proof $P$ (\Cref{def:good-pf-proof}),
$A$ on $P$ can be simulated by $O(n)$ calls to $T_1$, with context length bounded by $c_A L(L + D + 1)$ for some constant $c_A$ depending on $A$.
\end{theorem}

\begin{proof}
Fix a block-wise verification algorithm $A$ with module verifier
$V_{\mathrm{mod}}$. Let $P=(G,T)$ be a good Pseudo-Formal proof with node set
$V=\{v_1,\ldots,v_n\}$, where each node is written as
\[
    v=(P_v,c_v,\pi_v).
\]
By definition of block-wise verification, $A$ verifies $P$ by running
\[
    V_{\mathrm{mod}}(C_v(P))
\]
for every module $v\in V$, and accepts if and only if all such module checks
accept. Thus it suffices to simulate each call to $V_{\mathrm{mod}}$ by a
call to a fixed-size Transformer.

We first bound the context size of a single module check. For a node $v$, its
module context is
\[
    C_v(P)
    =
    \Bigl(
        (P_u)_{u\in \Anc_T(v) \cup \Dep_G(v)},
        (c_u)_{u\in \Dep_G(v)},
        P_v,c_v,\pi_v
    \Bigr),
\]
where
\[
    \Dep_G(v)=\{u\in V:(v,u)\in E_G\}.
\]
Since $P$ is a good Pseudo-Formal proof, the depth of $T$ is at most $D$, and
therefore $|\Anc_T(v)|\le D$. Moreover, the out-degree of every node in $G$ is
at most $L$ (at least one character is needed to specify one invoked lemma in the proof), so $|\Dep_G(v)|\le L$. Finally, the local size of each module is
bounded by $L$: the conclusion, the number and length of the premises, and the
local proof are all bounded by $L$. Hence the serialized length of $C_v(P)$ is
at most $L(D+L+1)$.

Now view the fixed module verifier $V_{\mathrm{mod}}$ as a Turing machine
$M_{\mathrm{mod}}$ which, on input $C_v(P)$, halts and outputs either
$\accept$ or $\reject$. Since $V_{\mathrm{mod}}$ is fixed as part of the
algorithm $A$, the corresponding Turing machine $M_{\mathrm{mod}}$ is also
fixed, independent of $P$, $n$, $L$, and $D$.

By Theorem~2 of \citet{merrill2023expressive}, for any Turing machine
$M$ which, on inputs of length $1+N$, runs for at most $t(N)$ steps, with
$t(N)$ polynomially bounded, there exists a fixed-size decoder-only projected
pre-norm Transformer with strict causal saturated attention which, on input
$x$, simulates $M(x)$ using $t(N)$ decoding steps and then outputs $M(x)$
using $|M(x)|$ additional decoding steps. 

Applying this theorem to
$M_{\mathrm{mod}}$, as we assume the block verifier has linear time complexity, there exists a fixed-size Transformer $T_1$ that simulates
one execution of $V_{\mathrm{mod}}$ on any module context $C_v(P)$. The size
of $T_1$ depends only on the fixed verifier $V_{\mathrm{mod}}$, and not on the
particular proof $P$.

Therefore, to simulate $A$ on $P$, we make one call to $T_1$ for each module
$v\in V$, with input $C_v(P)$. There are exactly $n$ modules, so this requires
$n=O(n)$ calls. Each call has context length bounded by the serialized length
of the corresponding module context, which is $O(L(D+L+1))$. The simulated algorithm accepts exactly
when all simulated module checks accept, which is precisely the acceptance
condition of $A$. This proves the theorem.
\end{proof}

\section{Ablations}

We ablate each stage of our method, results 
shown in \Cref{fig:pr_full_proof}. All
ablations run on the same 200 \texttt{Hard2Verify} problems (1,860 solution
steps), use \texttt{gpt-5.4-mini-2026-03-17} at \texttt{medium} reasoning effort,
and are scored with the same step-level metric as the main results: the positive
class is an \emph{incorrect step}, predictions from $k$ runs are pooled
pessimistically (a step is flagged if any of the $k$ runs flags it), and
TP/FP/FN are summed over all steps. Curves sweep $k$ 
values with a 500-iteration bootstrap over $k$-of-$n$ run
subsets. Every ablation is evaluated on the \emph{same} stored rollouts as the
full method wherever the ablated stage permits it, so the comparisons are paired
rather than merely distributional.

\paragraph{No faithfulness.} The faithfulness check is removed from
pseudo-formalization: each rollout still generates its own fresh PF, but
candidate rewrites are accepted if they pass a free, programmatic parsing check
alone rather than being audited block-by-block against the original solution and
regenerated on disagreement. Block verification and calibration are unchanged. We
resample until each problem has four PFs that pass the round-trip, and run one
block verification pass per PF, giving $k = 1 \dots 4$. This costs roughly a
quarter less than the full pipeline. Overall the performance decreases 
by a small amount, suggesting the faithfulness step buys a small 
improvement for a fairly large cost. It is at the discretion of the end 
user if they wish to pay this.

\paragraph{No BV.} Block verification is removed entirely. A single call receives
the problem, the indexed original steps and the PF, and must both \emph{find} the
errors and localise them to steps; there is no block-level error list. This
isolates the value of the pseudo-formalization itself, since the arm still sees
the structured rewrite but never verifies it piecewise. It remains above the
LLM-as-judge baseline at every operating point  but below the full pipeline. 

\paragraph{No Calibrator.} The calibrator's adjudication is removed while
retaining the mapping it performs. The arm receives exactly the same block-level
flags as the full pipeline, verified identical for all 1,600 (problem,
rollout) pairs, but is instructed to treat those flags as ground truth and
merely map them onto original steps, rather than validating each flag as genuine
or a rewrite artifact, independently re-grading every step, or reporting errors
the block verifier missed. The curve
falls \emph{below} the plain LLM-as-judge baseline. This suggests the LM 
calibration is an important component.

\paragraph{Stepwise full context (No PF).} Pseudo-formalization is removed but
per-unit decomposition is retained: one call per \emph{original step} rather than
per PF block. Each call sees the whole solution and is asked to judge one
designated step. This separates the two things the gap between the full pipeline
and the baseline otherwise conflates, the structured rewrite, and the act of
splitting verification into many small focused calls. It is roughly
compute-matched to block verification (9.3 steps per solution against 9.9 PF
blocks). The arm occupies a markedly more conservative regime, trading recall for
precision, and although its curve is flatter than the
baseline's it needs $k = 6$ (about 56 calls per solution) to reach a recall the
baseline attains with a single call. Parallel block calls alone therefore does not
account for the pipeline's advantage.

\paragraph{Stepwise single step (No PF).} As above, but each call sees only its
own step and the problem statement, no surrounding solution. It is strictly
worse than the full-context variant on both axes at every $k$ ($k = 4$:
$.730 / .610$), and by construction cannot mark a step incorrect on the grounds
that it inherits an earlier error. It is shown on its own axis scale in the
inset, as it lies far off the frontier occupied by every other arm.

\section{Prompts used in the Hard2Verify experiments}
\label{app:hard2verify-prompts}
\label{app:imo-prompts}

This appendix records the exact prompts used in the \emph{Hard2Verify}~\citep{pandit2025hard2verify}
experiments. Curly braces of the form \emph{\{problem\}} denote template fields
that are substituted at runtime with the corresponding problem statement,
stepwise solution, rewritten proof, error list, etc. We use \\
\texttt{gpt-5.4-mini-2026-03-17}
(medium reasoning effort) for every stage.

\subsection{Direct step-level verifier baseline}
\label{app:hard2verify-prompt-baseline}

The direct LLM-as-judge baseline uses the same step-level evaluation prompt in~\citet{pandit2025hard2verify}. The model receives the problem and the solution
split into indexed steps, then returns a list of correctness verdict for
each steps. We include the prompts here for completeness.

\subsubsection*{System prompt}

\begin{promptbox}
You are a strict, reliable math grader.
Return the evaluation using the exact format requested by the user.
Provide brief justifications (1-2 sentences per step); do not reveal chain-of-thought or scratch work.
\end{promptbox}

\subsubsection*{User prompt}

\begin{promptbox}
The following is a math problem and a solution (split into steps, enclosed with tags and indexed from 0):

[Math Problem]

{problem}

[Solution]

{steps}


    Your task is to review and critique the solution step-by-step.
    For each step, determine if it is correct or incorrect.
    - A correct step is one where all of the content is correct, and is logically consistent with all previous steps and information given in the problem.
    - An incorrect step is one where the content is incorrect, or is not logically consistent with all previous steps and information given in the problem, or is based on an error in a previous step.

    Important: Any step that contains or is based on an error is considered incorrect. That is, if the error is carried forward from a previous step or is based on an error in the previous step, consider the step incorrect.

    Provide reasoning for your correctness determinations. Your final verdict should be a comma-separated list of yes and no's, where each yes or no corresponds to a step's correctness, with yes meaning correct and no meaning incorrect.

    Please use the following format to return your answer:
    Reasoning: <your reasoning for each step>
    Verdict: <your comma-separated list of yes and no's>

    Do not use any other formatting, including markdown, bold text, code blocks, or any other formatting. If your formatting is incorrect, your evaluation will be affected.
\end{promptbox}

\subsection{Pseudo-Formal rewriter}
\label{app:hard2verify-prompt-rewriter}

PF+BV first rewrites each candidate solution into the structured format of
\Cref{def:pf-proof}. The output is XML-tagged and parsed deterministically into a
tree of propositions and lemmas.

\begin{promptbox}
Rewrite the following theorem and proof into a structured formal proof outline.

Structure:
- Use at most 2 layers in the proof tree:
  1. Propositions
  2. Lemmas
- Do not introduce any deeper hierarchy.
- If a deeper tree seems natural, flatten it into a sequential list of lemmas inside the relevant proposition.
- The theorem proof may only cite propositions; a proposition proof may only cite its own lemmas and the statements of earlier propositions.
- No trivial decompositions: do not decompose the theorem into a single proposition that restates the theorem, nor a proposition into a single lemma that restates the proposition.
- Every rewritten assertion must be unambiguous: definitions and statements must be precise, and each variable must have a single, clear meaning within an assertion.


Numbering and order:
- Number propositions sequentially: 1, 2, 3, ...
- Number lemmas within each proposition: 1.1, 1.2, ..., 2.1, 2.2, ...
- The proof must read top-to-bottom as a forward sequence: if component j uses component i at the same level, then j > i. Reorder to avoid forward references.

Faithfulness to the original proof:
- Preserve the proof's content, notation, logical flow, ordering, and wording as much as possible.
- For low-level arguments inside leaf components, avoid changing the original wording.
- Only make the minimal edits needed to fit the structured format.
- Do not introduce alternative arguments, and do not repair, optimize, strengthen, or silently fix the proof.
- Do not add justifications absent from the original proof, omit relevant proof details, or introduce statements stronger than what the original proof establishes.

Assumptions, conditions, and definitions:
- Clearly state the assumptions, conditions, and definitions for every theorem, proposition, and lemma.
- Each component may inherit the setting of its enclosing parent, but if tracing back through multiple earlier statements would be needed, restate the relevant assumptions explicitly.
- If a component modifies its parent's setting, explicitly state the full updated assumptions and note which assumptions were added, removed, or changed relative to the parent.

Output format:
Wrap every section in XML-style delimiter tags as shown in the template below.
- Use EXACTLY the tag names shown: THEOREM_STATEMENT, PROPOSITION_STATEMENT, LEMMA_STATEMENT, LEMMA_PROOF, PROPOSITION_PROOF, THEOREM_PROOF.
- Every tag except THEOREM_STATEMENT and THEOREM_PROOF MUST have an id attribute matching the numbering scheme above.
- Do NOT nest tags inside each other. All tags are at the top level.
- Do NOT include any text outside of tags.
- Include all assumptions, conditions, and definitions INSIDE the statement tag they belong to.
- Use LaTeX notation for all mathematical notations.

Template:

<THEOREM_STATEMENT>
Assumptions / Conditions / Definitions.
- ...
Statement :
...
</THEOREM_STATEMENT>

<PROPOSITION_STATEMENT id="1">
Assumptions / Conditions / Definitions.
- ...
Statement :
...
</PROPOSITION_STATEMENT>

<LEMMA_STATEMENT id="1.1">
Assumptions / Conditions / Definitions.
- ...
Statement :
...
</LEMMA_STATEMENT>

<LEMMA_PROOF id="1.1">
[proof.]
</LEMMA_PROOF>

<PROPOSITION_PROOF id="1">
[explain how Lemmas 1.1, 1.2 imply Proposition 1.]
</PROPOSITION_PROOF>

<PROPOSITION_STATEMENT id="2">
...
</PROPOSITION_STATEMENT>

...

<THEOREM_PROOF>
[explain how Propositions 1, 2, ... imply the theorem.]
</THEOREM_PROOF>


Now rewrite the following theorem and proof in this format:

[PASTE THEOREM AND PROOF HERE]
\end{promptbox}

\subsection{Faithfulness check}
\label{app:hard2verify-prompt-faithfulness}

After parsing the rewriter output, each component of the rewrite (every lemma,
proposition, and the theorem) is checked against the original proof for
faithfulness in its own call, using the same context and established-results
structure as the block verifier. Any component flagged UNFAITHFUL is sent back to
the rewriter regeneration step below.

\begin{promptbox}
You are an expert mathematician reviewing whether a component of a rewritten proof faithfully represents the corresponding part of the original proof.

You are given:
1. **Problem Statement**: The mathematical problem.
2. **Original Proof**: The full original proof.
3. **Contexts**: Statements from the rewritten proof that the Assertion inherits definitions or assumptions from (e.g., the theorem statement, the enclosing proposition statement). These are provided so you can understand the scope of the Assertion.
4. **Established Results**: Statements from the rewritten proof that have already been checked and can be assumed to be faithfully rewritten. You may use them as reference points.
5. **Assertion**: The specific rewritten statement whose faithfulness you must verify.
6. **Proposed Proof**: The rewritten proof of the Assertion.

Your task is to determine whether the Assertion and its Proposed Proof **faithfully represent** the corresponding part of the Original Proof.

Check for:
1. **Strengthened or weakened claims**: Does the Assertion claim more or less than the original proof establishes at the corresponding point?
2. **Omitted content**: Does the Proposed Proof drop a non-trivial argument that appears in the original?
3. **Added content**: Does the Proposed Proof introduce new arguments, repairs, or proof ideas not present in the original?
4. **Notation drift**: Are variables, functions, or definitions used differently than in the original?
5. **Misinterpretation**: Does the Assertion or Proposed Proof misunderstand the original's reasoning or logical structure?
6. **Scope errors**: Are assumptions incorrectly inherited, dropped, or added compared to the original?

Instructions:
- Do NOT judge whether the original proof is mathematically correct. Your sole task is faithfulness.
- Only flag changes that alter mathematical meaning. Cosmetic rephrasing is fine.
- Use the Contexts to understand what the Assertion is allowed to assume.
- Use the Established Results as anchors: if a prior component was faithful, you can compare the current component's references to it.

At the very end of your response, output your verdict as a JSON block. Do NOT write anything after the JSON block.

If FAITHFUL:
```json
{"verdict": "FAITHFUL", "error_description": null}
```

If UNFAITHFUL:
```json
{"verdict": "UNFAITHFUL", "error_description": "Identify the specific discrepancy: what the rewrite says vs. what the original says."}
```

**PROBLEM STATEMENT**
{problem}

**ORIGINAL PROOF**
{original_proof}

**CONTEXTS**
{contexts}

**ESTABLISHED RESULTS**
{established_results}

**ASSERTION**
{assertion}

**PROPOSED PROOF**
{proof}
\end{promptbox}

\subsection{Rewriter regeneration}
\label{app:hard2verify-prompt-regenerate}

Used when the faithfulness checker flags a component as unfaithful. It takes the
previous rewrite and the list of discrepancies, and produces a corrected rewrite in
the same format. The \emph{\{rewrite\_instructions\}} field is the
rewriter prompt above. The rewriter is called at most 3 times
(with the initial rewrite counting as one). An attempt whose output fails to parse
also counts toward this budget. If the rewrite is still flagged after the third
attempt and it still parses, we simply 
proceed to block verification.

\begin{promptbox}
You previously produced a structured rewrite of a mathematical proof, but a faithfulness checker identified discrepancies between your rewrite and the original proof.

Your task is to produce a NEW rewritten proof that fixes the identified issues while still following the structured format.

You will receive:
1. **Problem Statement**: The mathematical problem.
2. **Original Proof**: The student's original proof (ground truth — your rewrite must faithfully represent this).
3. **Previous Rewritten Proof**: Your earlier attempt, which contained faithfulness errors.
4. **Identified Errors**: The specific discrepancies flagged by the checker.

Requirements for the new rewrite:
- Fix every issue listed in the Identified Errors. For each issue, make sure the new rewrite no longer deviates from the original proof in that way.
- Do NOT introduce new discrepancies: do not strengthen/weaken claims, omit steps, add arguments, drift in notation, or misrepresent the logical structure relative to the Original Proof.
- Preserve the original proof's content, notation, logical flow, ordering, and wording as much as possible.
- Follow the SAME structured output format (XML-style tags, numbering, etc.) as the rewriting instructions below.

--- Original rewriting instructions (follow these for the output format) ---
{rewrite_instructions}
--- End of rewriting instructions ---

**PROBLEM STATEMENT**
{problem}

**ORIGINAL PROOF**
{original_proof}

**PREVIOUS REWRITTEN PROOF (contains errors)**
{previous_rewrite}

**IDENTIFIED ERRORS**
{errors}

Now produce the corrected rewritten proof. Output ONLY the rewritten proof using the XML-tag format — no commentary before or after.
\end{promptbox}

\subsection{Block verifier}
\label{app:hard2verify-prompt-bv}

PF+BV verifies each node of the Pseudo-Formal proof tree independently. The
verdict is parsed from the trailing JSON block.

\begin{promptbox}
You are an expert mathematical proof verifier.

Your task is to verify whether the proposed proof of a specific statement, called "Assertion", is correct.

You are given:
1. **Contexts**: A sequence of statements from which the Assertion may or may not inherit definitions, assumptions, or conditions. These are often the parent or ancestor statements of the Assertion, and can be the same as the global theorem. They are provided solely so you can understand the definitions and assumptions of the Assertion. They have NOT been verified and may be incorrect. Do not treat them as established truths, and do not verify them yourself. By default, you can inherit the assumption from the parent statements. The Assertion will specify which settings or assumptions it inherits from these contextual statements.
2. **Established Results**: Statements that have already been verified or can be assumed to be correct. You may assume all established results are correct and use them freely — do NOT re-verify them. The proof of the Assertion can invoke these results as long as the assumptions are properly justified and the definitions are consistent.
3. **Assertion**: The specific statement whose proof you must verify.
4. **Proposed Proof**: The proof of the Assertion to verify.

Instructions:
- Verify ONLY the proposed proof of the Assertion.
- Read the Assertion carefully and analyze the proof step by step.
- Identify any incorrect, unjustified, or logically invalid reasoning.
- Pay close attention to potentially confusing or ambiguous interpretations of concepts.
- When the proof references an established result, you may trust its conclusion, but you must verify that it is correctly applied:
    - Check that the result is used within its valid scope.
    - Explicitly identify the assumptions of the referenced result and confirm that each one is satisfied in the current context.
    - Verify that the definitions used in the invoked established results are the same as in the Assertion.
    - Detail which assumptions hold and why.
- For every term used in the proof, verify that its interpretation is unambiguous and consistent throughout. If a term is used with different meanings in different places, the proof is incorrect — do not guess or resolve the ambiguity yourself.

If the proof is INCORRECT, you must additionally classify the error into exactly **one** of the following three categories:

- **(a) Genuine gap** — a step is missing justification AND the gap cannot be locally repaired with a small amount of work. The step is not obviously true to an expert; closing the gap would require a substantive new argument. This is clearly negative.
- **(b) Acceptable elision** — a step is missing justification BUT the step/assertion itself is correct, and is either obvious to an expert reader or does not require explicit justification. The proof has skipped routine work that does not need to be spelled out.
- **(c) Local repair** — a step contains an incorrect justification, but the step is correct or correctable with a small amount of local work (e.g. fixing a sign, citing the right lemma, or rewriting one short paragraph). The high-level argument survives.

At the very end of your response, you MUST output your final verdict as a JSON block. Do NOT write anything after the JSON block.

If CORRECT:
```json
{"verdict": "CORRECT", "error_description": null, "error_class": null}
```

If INCORRECT:
```json
{"verdict": "INCORRECT", "error_description": "Identify the specific step that fails, state what it claims, and explain why it is wrong or unjustified.", "error_class": "a" | "b" | "c"}
```

**CONTEXTS**

{contexts}

**ESTABLISHED RESULTS**

{established_results}

**ASSERTION**

{assertion}

**PROPOSED PROOF**

{proof}
\end{promptbox}

\subsection{Calibration}
\label{app:hard2verify-prompt-calibrator}

PF+BV uses the calibrator to 
map block-level reports back 
to the benchmark's
step-level output format. The calibrator is called 
once per rollout, including
rollouts in which the block verifier flags no errors. 
The \emph{\{errors\}} field
lists every block the block verifier judged incorrect 
as follows:

\begin{promptbox}
[run {run_number} / {block_id}]
score: 0
error_class: {error_class}
{error_description}
\end{promptbox}

If no block
is flagged, \emph{\{errors\}} is instead the single sentence ``No potential errors
were flagged by the first 1 rewritten-proof verification run(s).'' The full calibrator prompt is:

\begin{promptbox}
You are a strict, reliable mathematical proof grader.
Your task is to evaluate the ORIGINAL solution to a math problem, step by step.

The structured rewrite and the list of potential errors are diagnostic aids only. They are not what you are grading. The rewrite may introduce omissions, distortions, false claims, or artificial dependencies that were not present in the original solution.

You must decide which original solution steps are mathematically correct or incorrect, and identify the first incorrect original step.

Input data:

1. Math Problem:
   The original problem statement.

2. Original Solution Steps:
   The solver's original solution split into indexed steps, written as:
   <step>[0] ...</step>
   <step>[1] ...</step>
   and so on.

3. Rewritten Proof:
   A structured rewrite of the original solution, decomposed into theorem / proposition / lemma / claim / fact blocks. This rewrite was produced automatically and may contain artifacts that are not present in the original solution.

4. Potential Errors:
   A list of potential errors flagged by automated verification of the rewritten proof. These may or may not be genuine errors in the original solution.

Evaluation process:

1. Validate each potential error.
   For each flagged issue, decide whether it is a genuine mathematical error in the ORIGINAL solution or a false alarm caused by the rewrite. Cross-check against the original solution steps. If the original solution contains the missing reasoning, a stronger equivalent argument, or a correct alternative route, do not count the flag as an original-solution error.

2. Inspect the original steps as needed.
   The potential errors are useful leads, but your final output must be based on the original solution steps. Check any original steps needed to identify the first genuine incorrect step. If you find that the first incorrect step was not covered by any potential error, report it in <additional_errors>.

3. Grade the original steps, not the rewrite.
   For each original step, determine whether that step is correct or incorrect.
   - A correct step is one where all mathematical content is correct and logically consistent with the problem and all previous correct steps.
   - An incorrect step is one that contains a mathematical error, an unjustified load-bearing claim, a logical inconsistency, a misapplied theorem, a false computation, or reasoning that depends on an earlier incorrect step.
   - If a step merely repeats or depends essentially on an earlier incorrect step, mark it incorrect as well.
   - Do not mark a step incorrect for harmless terseness, routine omitted algebra, stylistic issues, or rewrite artifacts.

4. Identify the first incorrect step.
   The first incorrect step is the smallest original step index whose content is mathematically incorrect or whose reasoning first depends on an error.
   If every original step is correct, output -1.

Important calibration rules:

- The original solution steps are authoritative for the final verdict.
- The rewritten proof and potential errors are evidence, not ground truth.
- Do not penalize the original solution for mistakes introduced only by the rewriter.
- Do not excuse a genuine gap in the original solution merely because the rewrite attempted to repair it.
- Focus on mathematical substance: incorrect proof steps, false claims, unjustified key lemmas, invalid dependencies, incorrect computations, and misapplied results.
- Ignore typos, formatting, wording preferences, and non-load-bearing presentation issues.
- Be strict but not pedantic: standard facts, routine algebra, and obvious intermediate manipulations need not be fully spelled out when they are mathematically valid and unambiguous.

Output requirements:

Return your final answer using exactly the XML-style structure below.
Do not use Markdown, bullet lists, code fences, or any text after the closing </calibration> tag.
Inside descriptions you may write LaTeX freely.

<calibration>
  <flag_audit>
    <flag>
      <source>Brief identifier of the potential error being audited, such as the rewritten block name or a short quote.</source>
      <status>genuine</status>
      <original_step>the original step index most directly affected, or -1 if no single step applies</original_step>
      <explanation>Brief explanation of why this is a genuine error in the original solution.</explanation>
    </flag>
    <flag>
      <source>...</source>
      <status>false_alarm</status>
      <original_step>-1</original_step>
      <explanation>Brief explanation of why this is only a rewrite artifact or otherwise not an error in the original solution.</explanation>
    </flag>
  </flag_audit>
  <additional_errors>
    <error>
      <original_step>step index</original_step>
      <description>Brief description of any genuine original-solution error not already covered by a potential flag.</description>
    </error>
  </additional_errors>
  <step_verdicts>yes,no,yes,...</step_verdicts>
  <first_incorrect_step>N</first_incorrect_step>
</calibration>

Rules for the parsed fields:

- <step_verdicts> must contain exactly {num_steps} comma-separated entries.
- Each entry must be exactly yes or no.
- yes means the corresponding original step is correct.
- no means the corresponding original step is incorrect.
- <first_incorrect_step> must be the first index whose step verdict is no.
- If every step verdict is yes, <first_incorrect_step> must be -1.
- If there are no potential flags to audit, output an empty <flag_audit> block.
- If there are no additional errors, output an empty <additional_errors> block.

MATH PROBLEM:
{problem}

ORIGINAL SOLUTION STEPS:
{steps}

REWRITTEN PROOF:
{rewritten_proof}

POTENTIAL ERRORS FROM REWRITTEN-PROOF VERIFICATION:
{errors}
\end{promptbox}

\section{Arxiv paper filtering}
\label{app:arxiv-paper-filtering}

  We describe the pipeline used to construct 
  \texttt{ArxivMathGradingBench}. 
  The pipeline has three stages: (i)
   retrieval and regex filtering, (ii) LLM-based classification 
   of revision comments, and (iii) source download.

  \textbf{Stage 1: Metadata retrieval and regex filtering.}
  We query the public arXiv API for all papers in the
  \texttt{math.*} categories submitted between January 1,
  2025 and July 30, 2025. For every entry returned we 
  extract the arXiv ID, version number, primary category,
  publication date, abstract, and revision comment.
  If revision comment is empty (paper is v1 or no comment)
  we skip, else we apply a regex filter to the 
  revision comment field. We retain only papers whose comment matches the case-insensitive regex
      \begin{center}
      \texttt{correct|errat|error|fix|mistake|bug|flaw|wrong|revised|amendment}.
      \end{center}
       We additionally require the comment to mention a formal mathematical result, via the case-insensitive regex
      \begin{center}
      \texttt{lemma|theorem|proposition}.
      \end{center}

  \textbf{Stage 2: LLM classification.}
  The regex filters are recall-oriented and admit many false positives (e.g., authors writing ``corrected typos in the proof of Theorem~3''). To remove these, we
  pass each comment to GPT 5.4 nano and ask it to classify the correction as one of three categories using the following prompt:

\begin{promptbox}
You are classifying arXiv paper revision comments. Based on the comment below, classify the correction as one of exactly three categories:

  - "major": A significant mathematical error was fixed — e.g. a flawed proof, wrong theorem statement, claim weakened/downgraded, incorrect bound, flawed logic or
   decomposition.
  - "minor": A small but genuine mathematical correction — e.g. a wrong constant, missing factor, fixed proof step.
  - "none": No real mathematical mistake was fixed — e.g. typos, formatting, added references, added content, improved exposition, LaTeX fixes.

Respond with ONLY one word: major, minor, or none.

Comment: {comments}

\end{promptbox}

  We retain papers classified as either \emph{major} or \emph{minor} and discard the rest. 

  At this stage we have 100 papers. We 
  then remove any papers with a non permissive license,
  or whose comment we determine to be a false positive.
  This leaves us with the 35 papers we use for the 
  final dataset.

  \paragraph{Stage 3: Source download.}
  For each retained paper we download the PDF  
  for the version before the correction.  

\section{Prompts used in the arXiv experiments}
\label{app:arxiv-prompts}

This appendix records the exact prompts used in the \texttt{ArxivMathGradingBench}
experiments. Curly braces of the form \{problem\} denote template fields that are
substituted at runtime with the rewritten paper, the
per-component error list, etc. We use \texttt{gpt-5.4-mini} (medium reasoning effort)
for every stage. The arxiv pipeline mirrors the \texttt{Hard2Verify} pipeline of
Appendix~\ref{app:imo-prompts}, with four differences: the rewriter, faithfulness
check and calibrator take the rendered PDF as input, the 
rewriter supports multiple
top-level theorems, the faithfulness check is a single call over the whole rewrite
rather than one call per component, and the calibrator outputs an XML list of error
locations keyed to PDF labels rather than per-step verdicts.

\subsection{LLM-as-judge baseline}

Used as the direct \emph{baseline} verifier on \texttt{ArxivMathGradingBench}: a single
call that receives the rendered PDF of the paper and returns an XML
list of error locations.

\begin{promptbox}
You are a mathematical referee reviewing a paper submitted to a peer-reviewed mathematics journal. You are given the paper as a rendered PDF (attached as a file). Read the paper as a human reader would, with all theorem/lemma/proposition numbers rendered, and inspect the notation and equations directly in the PDF.

Your task is to identify any mathematical errors in the paper that would require revision before publication. Focus on errors of mathematical substance: incorrect proofs, unjustified steps, false claims, gaps in reasoning, miscomputed quantities, misapplied theorems, and similar issues. Do not report typos, stylistic concerns, formatting issues, or notational preferences.

CRITICAL — how to identify locations:

For each error you find, identify the location using the rendered label EXACTLY as it appears in the PDF. For example: "Theorem 19", "Lemma 2.3", "Proposition 5.1", "Corollary 3.5", "Theorem B".

If the error is in an unlabelled passage (no Theorem/Lemma/Proposition number applies), name the surrounding section heading or use a short descriptive locator that a reader could find in the PDF (for example "Section 3.2, paragraph after Definition 2.1").

OUTPUT FORMAT:

After your analysis, output your final answer using the following XML-style format. Use one `<error>` block per error. Use the field tags exactly as shown. You may write LaTeX math (with raw backslashes) freely inside the `<description>` field; do not escape anything. Do not output anything after the closing `</errors>` tag.

<errors>
  <error>
    <location>Theorem 4.2</location>
    <description>Brief description of the mathematical error and why it is wrong.</description>
  </error>
  <error>
    <location>Lemma 5.1</location>
    <description>...</description>
  </error>
</errors>

If you find no mathematical errors, output an empty errors block:

<errors>
</errors>
\end{promptbox}

\subsection{Pseudo-Formal rewriter}

Used in the first stage of the PF pipeline on whole papers: takes the rendered PDF
 and rewrites the paper into a structured tree with multiple
top-level theorems, a globally-numbered shared pool of propositions, and lemmas
nested under each proposition.

\begin{promptbox}
Rewrite the mathematical paper below into a structured formal proof outline.

You are given the rendered PDF of the paper (attached as a file). Read the paper as a human would, inspect the precise notation, equations, and exact wording of each statement and proof directly in the PDF, and identify the rendered numerical or letter labels of every theorem/lemma/proposition/corollary/claim.

Structure:
- The top level of the rewrite is a list of **theorems**. Use one `<THEOREM_STATEMENT id="N">` block per top-level result the paper proves. Top-level results are the paper's flagship theorem(s) plus any independent secondary theorems / propositions / corollaries that the paper states as headline results (i.e. results that are not themselves used merely as intermediate steps for another result).
- Below the theorems, use at most 2 further layers in the proof tree:
  1. Propositions (shared pool, numbered globally across all theorems)
  2. Lemmas (numbered within each proposition)
- Do not introduce any deeper hierarchy. If a deeper tree seems natural, flatten it into a sequential list of lemmas inside the relevant proposition.
- Each `<THEOREM_PROOF id="N">` may cite any of the propositions; a `<PROPOSITION_PROOF id="K">` may cite its own lemmas and the statements of earlier propositions. A theorem proof must NOT cite another theorem unless that theorem appears earlier in the rewrite (i.e. has a smaller id) and the citation reflects an actual dependency in the paper.
- No trivial decompositions: do not decompose a theorem into a single proposition that restates the theorem, nor a proposition into a single lemma that restates the proposition.
- Every rewritten assertion must be unambiguous: definitions and statements must be precise, and each variable must have a single, clear meaning within an assertion.

Numbering and order:
- Number theorems sequentially: `<THEOREM_STATEMENT id="1">`, `<THEOREM_STATEMENT id="2">`, ...
  - Order theorems in the order they appear in the paper.
- Number propositions sequentially across the whole paper: 1, 2, 3, ... (a single shared pool, NOT per-theorem).
- Number lemmas within each proposition: 1.1, 1.2, ..., 2.1, 2.2, ...
- The proof must read top-to-bottom as a forward sequence: if component j uses component i at the same level, then j > i. Reorder to avoid forward references.

Faithfulness to the original paper:
- Preserve the paper's mathematical content, notation, logical flow, ordering, and wording as much as possible.
- For low-level arguments inside leaf components, avoid changing the original wording.
- Only make the minimal edits needed to fit the structured format.
- Do not introduce alternative arguments, and do not repair, optimize, strengthen, or silently fix the paper.
- Do not add justifications absent from the original paper, omit relevant proof details, or introduce statements stronger than what the original paper establishes.

Assumptions, conditions, and definitions:
- Clearly state the assumptions, conditions, and definitions for every theorem, proposition, and lemma.
- Each component may inherit the setting of its enclosing parent, but if tracing back through multiple earlier statements would be needed, restate the relevant assumptions explicitly.
- If a component modifies its parent's setting, explicitly state the full updated assumptions and note which assumptions were added, removed, or changed relative to the parent.

Output format:
Wrap every section in XML-style delimiter tags as shown in the template below.
- Use EXACTLY the tag names shown: THEOREM_STATEMENT, PROPOSITION_STATEMENT, LEMMA_STATEMENT, LEMMA_PROOF, PROPOSITION_PROOF, THEOREM_PROOF.
- Every tag MUST have an id attribute matching the numbering scheme above. (Including THEOREM_STATEMENT and THEOREM_PROOF — unlike other rewrite formats you may have seen.)
- Do NOT nest tags inside each other. All tags are at the top level.
- Do NOT include any text outside of tags.
- Include all assumptions, conditions, and definitions INSIDE the statement tag they belong to.
- Use LaTeX notation for all mathematical notations.

Template:

<THEOREM_STATEMENT id="1">
Assumptions / Conditions / Definitions.
- ...
Statement :
...
</THEOREM_STATEMENT>

<THEOREM_STATEMENT id="2">
Assumptions / Conditions / Definitions.
- ...
Statement :
...
</THEOREM_STATEMENT>

<PROPOSITION_STATEMENT id="1">
Assumptions / Conditions / Definitions.
- ...
Statement :
...
</PROPOSITION_STATEMENT>

<LEMMA_STATEMENT id="1.1">
Assumptions / Conditions / Definitions.
- ...
Statement :
...
</LEMMA_STATEMENT>

<LEMMA_PROOF id="1.1">
[proof.]
</LEMMA_PROOF>

<PROPOSITION_PROOF id="1">
[explain how Lemmas 1.1, 1.2 imply Proposition 1.]
</PROPOSITION_PROOF>

<PROPOSITION_STATEMENT id="2">
...
</PROPOSITION_STATEMENT>

...

<THEOREM_PROOF id="1">
[explain how Propositions 1, 3, ... imply Theorem 1.]
</THEOREM_PROOF>

<THEOREM_PROOF id="2">
[explain how Propositions 2, 4, ... imply Theorem 2.]
</THEOREM_PROOF>


Now rewrite the following paper in this format. The paper is provided as the attached PDF.
\end{promptbox}

\subsection{Rewriter regeneration}

Re-invoked when the faithfulness checker flags a component as UNFAITHFUL. It takes the
previous rewrite and the list of discrepancies, and produces a corrected rewrite in
the same format. The rewriter is called at most 5 times (with the initial rewrite
counting as one). An attempt whose output fails to parse also counts toward this
budget. If the rewrite is still flagged after the fifth attempt and it still parses,
we simply proceed to block verification; if no attempt parses, the run reports no
errors.

\begin{promptbox}
You previously produced a structured rewrite of a mathematical paper, but a faithfulness checker identified discrepancies between your rewrite and the original paper.

Your task is to produce a NEW rewritten paper that fixes the identified issues while still following the structured format.

You will receive:
1. The rendered PDF of the original paper (attached as a file — ground truth — your rewrite must faithfully represent this).
2. Your previous rewritten paper, which contained faithfulness errors.
3. The specific discrepancies flagged by the checker.

Requirements for the new rewrite:
- Fix every issue listed in the Identified Errors. For each issue, make sure the new rewrite no longer deviates from the original paper in that way.
- Do NOT introduce new discrepancies: do not strengthen/weaken claims, omit steps, add arguments, drift in notation, or misrepresent the logical structure relative to the original paper.
- Preserve the original paper's content, notation, logical flow, ordering, and wording as much as possible.
- Follow the SAME structured output format (XML-style tags, numbering, multiple `<THEOREM_STATEMENT id="N">` blocks at the top, etc.) as the rewriting instructions below.

--- Original rewriting instructions (follow these for the output format) ---
{rewrite_instructions}
--- End of rewriting instructions ---

PREVIOUS REWRITTEN PAPER (contains errors):
{previous_rewrite}

IDENTIFIED ERRORS:
{errors}

Now produce the corrected rewritten paper. Output ONLY the rewritten paper using the XML-tag format — no commentary before or after.
\end{promptbox}

\subsection{Faithfulness check}

After parsing the rewriter output, the rewritten paper is checked
against the original PDF for faithfulness in a single call covering every
component; any component flagged UNFAITHFUL is
sent back to the rewriter regeneration step above.

\begin{promptbox}
You are an expert mathematician reviewing whether a rewritten mathematical paper faithfully represents the original paper.

You are given:
1. **Original Paper (PDF)**: The rendered PDF of the original paper, attached as a file. This is the ground truth.
2. **Rewritten Paper**: The full structured rewrite, as a tree of theorems, propositions and lemmas. Each component carries an id.
3. **Component ids**: The complete list of ids you may refer to.

Your task is to determine, for **every component** of the Rewritten Paper, whether that component's statement and proof **faithfully represent** the corresponding part of the Original Paper.

Check for:
1. **Strengthened or weakened claims**: Does the component claim more or less than the original paper establishes at the corresponding point?
2. **Omitted content**: Does the component's proof drop a non-trivial argument that appears in the original paper?
3. **Added content**: Does the component's proof introduce new arguments, repairs, or proof ideas not present in the original paper?
4. **Notation drift**: Are variables, functions, or definitions used differently than in the original paper?
5. **Misinterpretation**: Does the component misunderstand the original's reasoning or logical structure?
6. **Scope errors**: Are assumptions incorrectly inherited, dropped, or added compared to the original?

Instructions:
- Do NOT judge whether the original paper is mathematically correct. Your sole task is faithfulness.
- Only flag changes that alter mathematical meaning. Cosmetic rephrasing is fine.
- A component may assume the statements of the components that enclose it; judge it within that scope.
- Work through the components in the order they are listed. A component you do not report is treated as faithful, so only report a component when you have identified a specific discrepancy.
- Use ONLY ids from the Component ids list. Do not invent ids.

OUTPUT FORMAT:

After your analysis, output your final answer using the following XML-style format. Use one `<component>` block per unfaithful component. Use the field tags exactly as shown. You may write LaTeX math (with raw backslashes) freely inside the `<description>` field; do not escape anything. Do not output anything after the closing `</unfaithful>` tag.

<unfaithful>
  <component>
    <id>an id from the Component ids list</id>
    <description>Identify the specific discrepancy: what the rewrite says vs. what the original paper says.</description>
  </component>
</unfaithful>

If every component is faithful, output an empty block:

<unfaithful>
</unfaithful>

**ORIGINAL PAPER**
The original paper is the PDF attached to this message.

**COMPONENT IDS**
{component_ids}

**REWRITTEN PAPER**
{rewritten_paper}
\end{promptbox}

\subsection{Block verifier}

Reused unchanged from the Hard2Verify pipeline; see Appendix~\ref{app:imo-prompts},
block-verifier section.

\subsection{Calibrator}

A second-pass referee called after the block verifier has produced per-component
error reports. The calibrator decides which flagged errors are genuine errors in
the underlying paper (vs.\ rewriting artifacts), may add errors the component
verifier missed, and emits the final XML error list with locations keyed to the
labels rendered in the PDF (e.g.\ Theorem~19'', Lemma~4.3''), not the
rewriter's internal numbering.

\begin{promptbox}
You are an expert mathematical referee. Your task is to produce the FINAL list of mathematical errors in a peer-reviewed mathematics paper.

You are given:
1. The rendered PDF of the original paper (attached as a file). This is your authoritative source for the paper's content — precise notation, equations, exact wording — and for the labels you must use in the final answer (e.g. "Theorem 19", "Lemma 4.3", "Proposition 5.1", "Theorem B").
2. A structured rewrite of the paper, decomposed into theorems / propositions / lemmas. The rewrite was produced by an automated rewriter and may itself introduce mistakes, omissions, or distortions that were **not** present in the original paper.
3. A list of potential errors that an automated component-verifier flagged in specific components (lemmas, propositions, or theorems) of the rewritten paper. These potential errors may or may not be genuine errors in the underlying paper — in particular, an "error" may be an artifact of the rewriting process (e.g., the rewriter dropped a key step, misstated a claim, or restructured the argument in a way that obscures correct reasoning that **is** present in the original paper).

Evaluation process:

1. **Error validation.** For each potential error, carefully determine whether it is a genuine error in the underlying paper or a false alarm. Examine the error in the context of the full rewritten paper, AND cross-check against the original paper (the attached PDF) to see whether the reasoning the rewritten paper is missing or misstating actually appears (correctly) in the original. If so, treat the error as a rewriting artifact rather than a genuine error in the paper.

2. **Search for additional errors.** You may also report errors in the original paper that the component verifier did NOT flag — including errors in unlabelled prose passages, between numbered results, or that span multiple components. Do not feel constrained to only the listed potential errors.

3. **Report only mathematical errors.** Focus on errors of mathematical substance (incorrect proofs, unjustified steps, false claims, gaps in reasoning, miscomputed quantities, misapplied theorems). Do not report typos, stylistic concerns, formatting issues, or notational preferences.

CRITICAL — labelling of locations in the final answer:

For every error you report, the `<location>` field MUST use the rendered label exactly as it appears IN THE PDF. Examples: "Theorem 19", "Lemma 4.3", "Proposition 5.1", "Corollary 3.5", "Theorem B".

Do NOT use:
- Internal rewrite-tree labels like "Proposition 7" or "Lemma 1.2" (those are the rewriter's invented numbering and do not match the paper's numbering).
- Raw LaTeX `\ref{...}` or `\label{...}` references.

If the error is in an unlabelled passage, use a short descriptive locator a reader can find in the PDF (for example "Section 3.2, paragraph after Definition 2.1").

OUTPUT FORMAT:

Output your final answer using the following XML-style format. Use one `<error>` block per error. Use the field tags exactly as shown. You may write LaTeX math (with raw backslashes) freely inside the `<description>` field; do not escape anything. Do not output anything after the closing `</errors>` tag.

<errors>
  <error>
    <location>Theorem 4.2</location>
    <description>Brief description of the mathematical error and why it is wrong.</description>
  </error>
  <error>
    <location>Lemma 5.1</location>
    <description>...</description>
  </error>
</errors>

If you find no genuine mathematical errors in the paper, output an empty errors block:

<errors>
</errors>

REWRITTEN PAPER (PROPOSED STRUCTURED FORM):
{rewritten_paper}

POTENTIAL ERRORS (from automated verification of the rewritten paper):
{errors}
\end{promptbox}

\subsection{\texttt{ArxivMathGradingBench} Judge}
\label{sec:app_arxiv_judge}

We use an LLM judge (GPT 5.4 mini
in the case of our experiments) 
to decide 
if the errors outputted by 
the verifier match those of 
the ground truth labels.
This judge receives 
the paper PDF, the ground truth 
error location(s) and the verifier 
predicted error locations(s). 
For each predicted error, it outputs
1 if the predicted location 
is the same as \emph{or contained 
within} the ground truth 
error location. For example, 
if the ground truth is ``Theorem 1'' 
and the prediction is ``theorem 1'', 
this is labeled as correct 
(ignoring the case differences). 
If the prediction is ``Proposition 1.1'',
which is \emph{within Theorem 1} in 
the paper PDF, then this is also 
labeled as correct.

We note there is a possibility of 
labeling error here. Perhaps 
the error the author noted in
Theorem 1 was actually Proposition 1.2,
not 1.1. The other labeling option
is to only match if the predicted error 
and ground truth are at the same 
level of the proof hierarchy. That is 
in the prior case, only predicting 
variants of ``Theorem 1'' would be 
judged as correct. We call 
this the \emph{Strict} grading, 
and the prior method \emph{Inclusive} 
grading. The results reported in 
the main paper use Inclusive grading.

\begin{figure}[h]
    \centering
\includegraphics[width=1.0\textwidth]{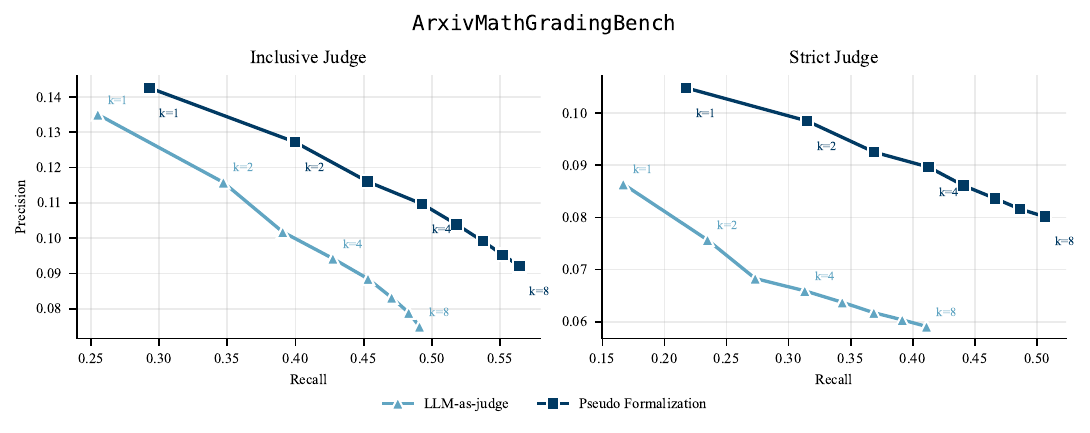}
    \caption{\texttt{ArxivMathGradingBench} 
    35 problem subset performance using two different 
    LLM judges for grading (independent 
    runs from those reported in the paper). We report 
    Inclusive in the main paper. For Strict, the gain of PF over baseline is even larger.} 
    
    \label{fig:figure_strict}
\end{figure}

We test using Strict grading,
and actually find Pseudo-Formalization
outperforms the baseline by even 
larger margins than when using 
Inclusive grading (see \Cref{fig:figure_strict}). We 
believe Inclusive is the better 
grading strategy however, and 
so report performance using the 
inclusive metric in the main body.

The prompt for the Inclusive 
grader is below:

\begin{promptbox}
You must decide whether a verifier's predicted error refers to a known ground-truth (GT) error location in a mathematics paper. The full paper is attached as a PDF — use it to resolve where the predicted location actually sits in the paper's structure.

Paper title: {title}

Ground-truth error locations (the author corrected errors at these locations):
{gt_list}

Predicted error:
- Location: {pred_loc}
- Description: {pred_desc}

A predicted error matches a GT location ONLY if one of these holds:
(a) its location names the same result — same number/letter label (e.g. "proof of Theorem 19" or "Theorem 19, equation (4.2)" match GT "Theorem 19"); or
(b) in the attached paper, its location is a specific step, claim, case, equation, or passage that lies INSIDE the GT result's statement or proof (e.g. "Claim 2.2 in the proof of Theorem 2.1" matches GT "Theorem 2.1").

It is NOT a match if the predicted location names a DIFFERENT numbered result or a different part of the paper that merely uses, cites, depends on, or inherits a flaw from the GT result — even if the description discusses the GT result extensively. Different numbers never match ("Theorem 18" does not match GT "Theorem 19").

Answer with ONLY a JSON object, no other text:
{"matched_gt_indices": [<0-based indices into the GT list, empty if none>]} 
\end{promptbox}

The prompt for the Strict 
grader is below:

\begin{promptbox}
    You must decide whether a verifier's predicted error refers to a known ground-truth (GT) error location in a mathematics
     paper. You are only given the location strings, which is sufficient: this is a strict label comparison.

     Paper title: {title}

     Ground-truth error locations (the author corrected errors at these locations):
     {gt_list}

     Predicted error:
     - Location: {pred_loc}
     - Description: {pred_desc}

     A predicted error matches a GT location ONLY if its location names the SAME result — the same number/letter label, or an explicitly stated alias (e.g.
     "Theorem 4.24 (Theorem B)" matches GT "Theorem B"; "proof of Theorem 19" or "Theorem 19, equation (4.2)" match GT "Theorem 19").

     It is NOT a match if the predicted location names anything else: a different numbered result, an unnumbered claim/step/case inside the GT result's proof,
     a section or passage, or a result that merely uses, cites, depends on, or inherits a flaw from the GT result — even if the description discusses the GT
     result extensively. Different numbers never match ("Theorem 18" does not match GT "Theorem 19").

     Answer with ONLY a JSON object, no other text:
     {"matched_gt_indices": [<0-based indices into the GT list, empty if none>]}
\end{promptbox}

\section{Prompts used in the First Proof experiments}
\label{app:firstproof-prompts}

We show the prompts used in the
\texttt{FirstProofGradingBench} experiments. We use \texttt{gpt-5.5} at high
reasoning effort for every stage.
The pipeline differs from the arXiv pipeline of
Appendix~\ref{app:arxiv-prompts} in three ways. Methods are
preceded by a theorem check. The block verifier has a different prompt.
The calibrator has a different prompt.

\subsection{LLM-as-judge baseline}
A single
call receives the problem statement and the submitted proof and returns a list
of errors.

\begin{promptbox}
You are a mathematical referee reviewing a paper submitted to a peer-reviewed mathematics journal, and YOU HAVE WEB SEARCH. You are given the paper's raw LaTeX source (included below). It is a single self-contained solution to a research problem: one main result with its proof. Read it as a referee would, inspecting the notation and equations directly in the source.

Your task is to identify any mathematical errors in the paper that would require revision before publication. Focus on errors of mathematical substance: incorrect proofs, unjustified steps, false claims, gaps in reasoning, miscomputed quantities, misapplied theorems, and similar issues. Do not report typos, stylistic concerns, formatting issues, or notational preferences.

You must be MAXIMALLY RIGOROUS about the mathematics: verify every load-bearing step, and give no benefit of the doubt to cited external results. Carry out ALL THREE steps below before deciding what to report.

STEP 1 — DEFINITION PINNING (object, not name). List every nontrivial term, object, operator, or named notion in the paper's main result and its proof. For EACH, it is acceptable if:
  (a) it is defined in the paper itself — an explicit formula, a defining role (e.g. a normalizing constant, a mollifier, the CDF of a stated density), or consistent unambiguous usage all count as definitions; OR
  (b) you retrieve a definition of the OBJECT from a credible source (textbook, peer-reviewed paper, established reference work) via web search. Quote the definition and cite the source.
OBJECT-PINNING RULE: pin the mathematical OBJECT, not its name. Notation or naming that differs from an external source's convention is NOT a defect. If the paper admits one coherent reading under which every use of the term is correct, adopt that reading and do not flag it. Report an undefined or ambiguous term as an error ONLY if the ambiguity infects the mathematics — some load-bearing step fails under EVERY reasonable reading consistent with the paper — and then name the offending term and the step it breaks. If the paper's usage genuinely conflicts with the definition the paper fixes for a term, that is also an error.

STEP 2 — LEMMA / CITED-RESULT PINNING (verbatim hypotheses AND conclusion). For every cited or invoked external result, retrieve its VERBATIM statement — BOTH its full hypotheses and its conclusion — from a credible source, and quote it. Then go through the hypotheses ONE ASSUMPTION AT A TIME and decide, for EACH SINGLE assumption, whether it genuinely holds in the current context — list the assumption, state holds/fails, and justify. Only if every single assumption fits may the result be applied. Also confirm that the conclusion the proof uses is exactly the conclusion the source states (not a stronger or broader version). Be especially alert to OVERGENERALIZATION (a result invoked in greater generality than the source actually establishes it) and to DOMAIN / OBJECT MISMATCH (the cited result is about a different class of objects, structure, or setting than the one at hand). A nonexistent, not-in-source, misstated, overgeneralized, mismatched, or otherwise misapplied result => report it as an error, naming the exact failing assumption or the precise mismatch between what the source states and what the proof uses.

STEP 3 — SEEK A COUNTEREXAMPLE. Actively try to REFUTE the paper's main result or a load-bearing claim in its proof, using whatever methods are appropriate. If you find a valid counterexample, report it as an error and state it explicitly.

Default-to-error rule (cited results): if after honest effort you cannot confirm every hypothesis of a cited or invoked external result, do NOT give the benefit of the doubt — report it as an error and say which check failed. For definitions and terminology, apply the OBJECT-PINNING RULE above instead. Continue to allow routine, non-load-bearing steps to be filled: terseness, skipped arithmetic, standard manipulations, minor repairable slips, and standard background facts of the ambient theory that an expert referee would grant without citation (e.g. that a named contraction semigroup is a contraction, or that a normalizing constant normalizes) — provided the fact is not itself the crux of the paper's main result.

GAP-FILLING DISCIPLINE. Do NOT report a step as an error merely because it omits intermediate justification, nor because it contains a local slip that a careful reader can repair on the spot. Your job is to detect genuine errors — load-bearing false claims, misapplied results, logical invalidity, inconsistent use of terms — NOT to demand that every step be spelled out. Terseness, skipped arithmetic, standard manipulations, routine verifications, and minor slips are not errors when the intended correct statement is unambiguous from context and the rest of the proof still goes through; multiple such gaps do not compound into one. Before reporting, try to fill the gap or repair the slip yourself using the paper and standard mathematical knowledge appropriate to the problem's level. Report only when (a) the mistake is load-bearing (the paper's conclusion or a later step genuinely depends on the incorrect claim), or (b) the repair would require a substantive new idea, a nontrivial result, or a definition not available in the paper. A repair must not change the main result's hypotheses or conclusion: if fixing the proof would require adding an assumption, restricting the domain, or weakening the target, that IS an error.

ERROR REPORTING. Each <description> is a referee finding, not an audit transcript. Do not narrate your procedure ("Step 1 ...", "Step 2 ..."), do not list terms you successfully pinned, and do not comment on parts of the paper that are fine — do that work in your reasoning and keep it out of the description. Report EVERY independent defect you established, not only the first or most decisive one, each as its own <error> block. State each affirmatively: the exact step it breaks, the specific failing object, hypothesis, or quantifier (e.g. "the cited result gives this only for \(\mu\)-a.e. x, while the proof needs it for every x"), and why it fails on the merits.

Do NOT consult First Proof materials (1stproof.org or its GitHub); rely only on independent literature.

CRITICAL — how to identify locations:

For each error, copy into the <quote> field a VERBATIM span (character for character) from the paper source that contains the error — a contiguous span of one or a few sentences / one display. Do not paraphrase, do not add ellipses, do not fix typos. In the <location> field, additionally give a short human-readable locator (the surrounding section heading, environment, or a brief description such as "the lemma bounding the overlap integral").

OUTPUT FORMAT:

After your analysis, output your final answer using the following XML-style format. Use one `<error>` block per error. Use the field tags exactly as shown. You may write LaTeX math (with raw backslashes) freely inside the `<quote>` and `<description>` fields; do not escape anything. Do not output anything after the closing `</errors>` tag.

<errors>
  <error>
    <location>Section 3, the lemma bounding the overlap integral</location>
    <quote>exact text copied character-for-character from the source</quote>
    <description>Brief description of the mathematical error and why it is wrong.</description>
  </error>
  <error>
    <location>...</location>
    <quote>...</quote>
    <description>...</description>
  </error>
</errors>

If you find no mathematical errors, output an empty errors block:

<errors>
</errors>


**PAPER LATEX SOURCE**

{paper_tex}
\end{promptbox}

\subsection{Theorem check (baseline)}
Both methods are preceded by a theorem check, which asks whether the submission
proves the problem as posed or only a weaker, conditional, or different
statement. This step has no web access. In the baseline arm it receives the
original problem and the entire submission.

\begin{promptbox}
You are verifying whether a submitted proof actually proves the ORIGINAL problem as posed, or only a weaker / conditional / different statement.

You are given:
1. ORIGINAL PROBLEM — the exact statement that must be proved.
2. SUBMISSION — the submission's complete write-up: its claimed theorem(s) and all proofs (what it states and proves may differ from the original).

Determine whether the SUBMISSION establishes the ORIGINAL PROBLEM. Return INCORRECT if it:
- adds an extra hypothesis/assumption not in the original (i.e. proves only a conditional statement),
- proves a weaker conclusion, restricts to a sub-case, or handles only one branch of a required dichotomy,
- proves a different statement than asked, or
- leaves the original's required conclusion unproved.
Return CORRECT only if the submission fully establishes the ORIGINAL PROBLEM as stated. Do not penalise terseness or routine gap-filling within an otherwise-complete argument; focus solely on SCOPE — is the thing proved the thing asked?

End your response with exactly:
<verdict>CORRECT or INCORRECT</verdict>
<error_description>If INCORRECT: state precisely what was proved instead of the original (the extra assumption, the missing case, or the weakened conclusion). If CORRECT: leave empty.</error_description>

=== ORIGINAL PROBLEM ===
{original}

=== SUBMISSION ===
{proof}
\end{promptbox}

\subsection{Theorem check (Pseudo-Formalization)}
The same step in the Pseudo-Formalization pipeline. It receives the theorem block of
the Pseudo-Formal rewrite in place of the whole submission.

\begin{promptbox}
You are verifying whether a submitted proof actually proves the ORIGINAL problem as posed, or only a weaker / conditional / different statement.

You are given:
1. ORIGINAL PROBLEM — the exact statement that must be proved.
2. SUBMISSION'S CLAIMED THEOREM — what the submission states it proves (may differ from the original).
3. ESTABLISHED RESULTS — supporting propositions; assume these are TRUE (do not re-verify them).
4. PROOF — the submission's proof of its theorem.

Determine whether the PROOF establishes the ORIGINAL PROBLEM. Return INCORRECT if the proof:
- adds an extra hypothesis/assumption not in the original (i.e. proves only a conditional statement),
- proves a weaker conclusion, restricts to a sub-case, or handles only one branch of a required dichotomy,
- proves a different statement than asked, or
- leaves the original's required conclusion unproved.
Return CORRECT only if the proof fully establishes the ORIGINAL PROBLEM as stated. Do not penalise terseness or routine gap-filling within an otherwise-complete argument; focus solely on SCOPE — is the thing proved the thing asked?

End your response with exactly:
<verdict>CORRECT or INCORRECT</verdict>
<error_description>If INCORRECT: state precisely what was proved instead of the original (the extra assumption, the missing case, or the weakened conclusion). If CORRECT: leave empty.</error_description>

=== ORIGINAL PROBLEM ===
{original}

=== SUBMISSION'S CLAIMED THEOREM ===
{claimed}

=== ESTABLISHED RESULTS (assume true) ===
{established}

=== PROOF ===
{proof}
\end{promptbox}

\subsection{Pseudo-Formal rewriter}

\begin{promptbox}
Rewrite the mathematical paper below into a structured formal proof outline.

You are given:
1. The rendered PDF of the paper (attached as a separate file). Use it to read the paper as a human would and to identify the rendered numerical/letter labels of every theorem/lemma/proposition/corollary/claim.
2. The raw LaTeX source of the paper (included below). Use it for precise notation, equations, and the exact wording of statements and proofs.

Structure:
- The TOP level of the rewrite is a list of **theorems**. Use one ``<THEOREM_STATEMENT id="N">`` block per top-level result the paper announces (its flagship theorem(s) plus any independent secondary theorems / propositions / corollaries that the paper states as headline results — i.e. results that are not used merely as intermediate steps for another result). Number theorems sequentially: id="1", id="2", ... in the order they appear in the paper.
- Below the theorems, propositions form a single SHARED pool numbered globally (1, 2, 3, ...) — they are NOT scoped to a specific theorem; any theorem proof may cite any proposition from this pool.
- The proof tree has up to 5 levels, in this fixed order from the top: (1) Theorem -> (2) Proposition -> (3) Lemma -> (4) Claim -> (5) Fact. Always decompose using **consecutive levels starting from the top**. If you use 3 levels, use Theorem + Proposition + Lemma (NOT Theorem + Proposition + Claim — that skips Lemma). If you use 4 levels, use Theorem + Proposition + Lemma + Claim. If you use the maximum 5 levels, use all of Theorem + Proposition + Lemma + Claim + Fact. NEVER skip an intermediate level.
- Within the depth you have chosen, individual blocks at the deepest used level can be leaves (statement + proof, no sub-blocks). For example, when using 4 levels, some claims can be leaves (no facts under them) while other claims have facts. The constraint is only on which TYPES of blocks may appear, not on whether every branch reaches the deepest type.
- The dotted-id depth MUST match the tag: 1 dot-part -> PROPOSITION_STATEMENT (id="K"), 2 parts -> LEMMA_STATEMENT (id="K.L"), 3 parts -> CLAIM_STATEMENT (id="K.L.C"), 4 parts -> FACT_STATEMENT (id="K.L.C.F"). Never use a tag with a dotted-id of the wrong depth (e.g. do NOT emit `<CLAIM_STATEMENT id="2.1">` to mean "claim 1 of proposition 2" — that conflicts with lemma id "2.1"). If you want a sub-block of a proposition without an intermediate lemma layer, do not decompose at all and keep the argument inline in the PROPOSITION_PROOF.
- Do not introduce any deeper hierarchy than 4 levels below the theorem.
- If a deeper tree seems natural, flatten it into a sequential list of facts inside the relevant claim.
- Citation scope:
    - A non-theorem proof block (proposition / lemma / claim / fact) may cite only (a) its own direct children (the level immediately below it), and (b) any block declared earlier in the document that lies outside the current block's subtree, **excluding the THEOREM_STATEMENT blocks** (a non-theorem block may not cite a theorem statement).
    - A theorem proof (THEOREM_PROOF id="N") may cite (a) any proposition from the shared pool, and (b) any earlier theorem (id < N) by its statement.
    - A block must NOT cite its own ancestors, its descendants beyond direct children, any block from a later-declared subtree, or itself.
- Citations are statement-only: only the cited block's statement is available as a premise — its proof body is not. If you need an object or intermediate result from inside another block's *_PROOF, hoist it into its own block's statement and cite that.
- Dependency declaration: after every *_PROOF block, emit a ``<DEPS id="...">`` tag listing the ids of all blocks actually cited in that proof, comma-separated. Every id referenced in the proof text must appear in DEPS, and every id in DEPS must obey the citation scope above. If the proof cites nothing, emit an empty tag: ``<DEPS id="..."></DEPS>``.
    - For non-theorem proofs, the DEPS id is the SAME id as the proof it follows (e.g. ``<DEPS id="1.1.1">`` for ``<CLAIM_PROOF id="1.1.1">``). The body lists comma-separated dotted ids of cited blocks (no theorem references).
    - For theorem proofs, the DEPS id is ``"theorem_N"`` where N is the theorem id (e.g. ``<DEPS id="theorem_2">`` for ``<THEOREM_PROOF id="2">``). The body may contain (a) bare integer ids for cited propositions (and their descendants by dotted id, if any are needed), and (b) ``"theorem_M"`` references for any earlier theorem (M < N) the proof relies on. Example: ``<DEPS id="theorem_2">1, 3, theorem_1</DEPS>`` means Theorem 2's proof cites Propositions 1 and 3, plus Theorem 1's statement. This naming convention avoids collisions with proposition DEPS.
- No trivial decompositions: do not decompose a theorem into a single proposition that restates the theorem, nor a proposition into a single lemma that restates the proposition, nor a lemma into a single claim that restates the lemma, nor a claim into a single fact that restates the claim.
- If a theorem proof contains several distinct assertions or sub-arguments, decompose it into multiple propositions.
- If a proposition proof contains several distinct assertions or sub-arguments, decompose it into multiple lemmas.
- If a lemma proof contains several distinct assertions or sub-arguments, decompose it into multiple claims.
- If a claim proof contains several distinct assertions or sub-arguments, decompose it into multiple facts.
- For challenging proofs, decompose based on the nature of the argument: arguments of different kinds (combinatorial vs. algebraic vs. inductive vs. case analysis) should be separated into different claims or facts; and if only a subset of the constraints/hypotheses is needed to justify a step in the proof block, that step probably should be a separate unit.
- A *_PROOF block must not restate or paraphrase its own statement. The core assertion of a block goes in its statement; the *_PROOF holds only the justification, method, or new setup that is not already in the statement. If the only thing you would write in *_PROOF is a paraphrase of the statement, write exactly None instead — "None" is legitimate proof.
- Preserve every external citation from the original paper (references to prior papers, named theorems, etc.). At the end of each *_PROOF block that uses such citations, list those external results. If the original paper provides a proper citation with author name, link, and journal, provide them as well. External citations are separate from DEPS, which only tracks internal block-to-block references.

Numbering and order:
- Number theorems sequentially: id="1", id="2", ... in document order.
- Number propositions sequentially across the whole paper (single shared pool, NOT per-theorem): 1, 2, 3, ...
- Number lemmas within each proposition: 1.1, 1.2, ..., 2.1, 2.2, ...
- Number claims within each lemma: 1.1.1, 1.1.2, ... 1.2.1, 1.2.2, ... 2.1.1, 2.1.2
- Number facts within each claim: 1.1.1.1, 1.1.1.2, ... 1.1.2.1, 1.1.2.2, ... 2.1.1.1, 2.1.1.2, ...
- The proof must read top-to-bottom as a forward sequence: if component j uses component i at the same level, then j > i. Reorder to avoid forward references.

Faithfulness to the original paper:
- Preserve the paper's mathematical content, notation, logical flow, ordering, and wording as much as possible.
- For leaf components (facts, claims with no sub-facts and lemmas with no sub-claims), keep the original wording verbatim wherever possible; only edit when required by the structural constraints above.
- Only make the minimal edits needed to fit the structured format.
- Do not introduce alternative arguments, and do not repair, optimize, strengthen, or silently fix the paper.
- Do not add justifications absent from the original paper, omit relevant proof details, or introduce statements stronger than what the original paper establishes.
- Preserve the original's sequential reasoning. If the original derives steps in order, do not merge them into a joint or parallel deduction — that can hide errors that only appear step by step.

Assumptions, conditions, and definitions:
- Clearly state the assumptions, conditions, and definitions for every theorem, proposition, lemma, claim and fact.
- Each component may inherit the setting of its enclosing parent, but if tracing back through multiple earlier statements would be needed, restate the relevant assumptions explicitly.
- If a component modifies its parent's setting, explicitly state the full updated assumptions and note which assumptions were added, removed, or changed relative to the parent.

Output format:
Wrap every section in XML-style delimiter tags as shown in the template below.
- Use EXACTLY the tag names shown: THEOREM_STATEMENT, PROPOSITION_STATEMENT, LEMMA_STATEMENT, CLAIM_STATEMENT, FACT_STATEMENT, FACT_PROOF, CLAIM_PROOF, LEMMA_PROOF, PROPOSITION_PROOF, THEOREM_PROOF, DEPS.
- Every tag MUST have an id attribute matching the numbering scheme above (including THEOREM_STATEMENT and THEOREM_PROOF — multi-theorem mode). DEPS for THEOREM_PROOF carries id="theorem_N" as noted above.
- Do NOT nest tags inside each other. All tags are at the top level.
- Do NOT include any text outside of tags.
- Include all assumptions, conditions, and definitions INSIDE the statement tag they belong to.
- Use LaTeX notation for all mathematical notations.

Template (single-theorem case shown for compactness; emit additional ``<THEOREM_STATEMENT id="N">`` and ``<THEOREM_PROOF id="N">`` blocks when the paper has multiple top-level theorems):

<THEOREM_STATEMENT id="1">
Assumptions / Conditions / Definitions.
- ...
Statement :
...
</THEOREM_STATEMENT>

<THEOREM_STATEMENT id="2">
...
</THEOREM_STATEMENT>

<PROPOSITION_STATEMENT id="1">
Assumptions / Conditions / Definitions.
- ...
Statement :
...
</PROPOSITION_STATEMENT>

<LEMMA_STATEMENT id="1.1">
Assumptions / Conditions / Definitions.
- ...
Statement :
...
</LEMMA_STATEMENT>

<CLAIM_STATEMENT id="1.1.1">
Assumptions / Conditions / Definitions.
- ...
Statement :
...
</CLAIM_STATEMENT>

<FACT_STATEMENT id="1.1.1.1">
Assumptions / Conditions / Definitions.
- ...
Statement :
...
</FACT_STATEMENT>

<FACT_PROOF id="1.1.1.1">
[proof.]
</FACT_PROOF>
<DEPS id="1.1.1.1">[comma-separated ids of blocks cited in this fact proof, or empty]</DEPS>

<CLAIM_PROOF id="1.1.1">
[explain how Facts 1.1.1.1, 1.1.1.2, ... imply Claim 1.1.1.]
</CLAIM_PROOF>
<DEPS id="1.1.1">1.1.1.1, 1.1.1.2, ...</DEPS>

<LEMMA_PROOF id="1.1">
[explain how Claims 1.1.1, 1.1.2, ... imply Lemma 1.1.]
</LEMMA_PROOF>
<DEPS id="1.1">1.1.1, 1.1.2, ...</DEPS>

<PROPOSITION_PROOF id="1">
[explain how Lemmas 1.1, 1.2 imply Proposition 1.]
</PROPOSITION_PROOF>
<DEPS id="1">1.1, 1.2, ...</DEPS>

<PROPOSITION_STATEMENT id="2">
...
</PROPOSITION_STATEMENT>

...

<THEOREM_PROOF id="1">
[explain how Propositions 1, 3, ... imply Theorem 1.]
</THEOREM_PROOF>
<DEPS id="theorem_1">1, 3, ...</DEPS>

<THEOREM_PROOF id="2">
[explain how Propositions 2, 4, ... imply Theorem 2.]
</THEOREM_PROOF>
<DEPS id="theorem_2">2, 4, ...</DEPS>




============================== 
STRICTNESS REQUIREMENTS (OVERRIDING) 
==============================
Faithfulness is paramount. Obey ALL of the following; they override any temptation to clean up the argument:

1. ADD NOTHING. Do not insert any reasoning, justification, connective phrase, algebraic step, or intermediate result that is not literally present in the original. In particular, never write sentences like "Thus X is counted by Y", "equating gives ...", or "since p<1, p^2<p" unless that exact justification appears in the original. If the original merely ASSERTS a step with no justification, your corresponding *_PROOF must also merely assert it (or be exactly "None"). Do NOT paper over gaps.

2. DROP NOTHING. Every sentence of the original must be represented somewhere in the output. This includes material that is not a formal proof step:
   - worked EXAMPLES (render each as its own leaf block, e.g. a FACT or CLAIM whose statement is the example and whose proof is the example's verification, verbatim);
   - REMARKS and side comments;
   - any "what is missing", "open issues", "limitations", "final conclusion", or self-critical/disclaimer prose;
   - any statement that the proof is partial, conditional, heuristic, or proves a different/relaxed result than asked.
   Preserve such disclaimers verbatim in the most relevant block's statement or as a dedicated remark/fact block. Do NOT silently omit them.

3. DO NOT REFRAME. Keep the original's own stance. If the original asserts a (false or gappy) claim affirmatively, state it affirmatively; do not soften it to "this is not established" and do not strengthen a hedge into a claim.

4. If faithfully preserving a piece of prose conflicts with the structural rules above, prefer FAITHFULNESS: attach the prose to the nearest block rather than dropping it.
============================================


==============================
FORMAT REQUIREMENTS (OVERRIDING) 
==============================
Keep STATEMENT and PROOF strictly separated in EVERY block:
1. A *_STATEMENT contains ONLY the assumptions/conditions/definitions and the precise claim. It MUST NOT contain any derivation, computation chain, justification, or commentary (e.g. "No commutativity was used because X = Y", "perimeter = 2\sqrt3", "Being flat, it self-intersects", "since X is integer-valued ..."). All such content belongs in the *_PROOF.
2. A *_PROOF must contain the ACTUAL justification. It MUST NOT merely restate the claim (e.g. just "Therefore <the statement>"); if the only content would be a restatement, the real argument has been misplaced — put it here instead, or write exactly None.
3. Emit NO leftover scaffolding: no "Original proof label: ...", no duplicated "Assumptions / Conditions / Definitions" header inside a block, no document preamble (\title, \author, \maketitle, Disclaimer, Problem Statement headers) inside any block.
4. A STATEMENT must stand alone: do not begin it with a bare connective ("Then", "Therefore", "Hence", "Thus", "Consequently") that refers to text outside the block; restate the needed antecedent in the assumptions.
Faithfulness still holds: do not add or drop mathematical content; only place each piece in the correct field.
============================================

Now rewrite the following proof in this format. No PDF is attached; the raw LaTeX source below is the complete and authoritative source. This is a single self-contained theorem with its proof, so expect exactly one THEOREM_STATEMENT.

LATEX SOURCE:
{paper_tex}
\end{promptbox}

\subsection{Rewriter regeneration}
Re-invoked when a check below flags the rewrite. It takes the previous rewrite
together with the specific discrepancies and produces a corrected rewrite in the
same format.

\begin{promptbox}
You previously produced a structured rewrite of a mathematical paper, but a faithfulness checker identified discrepancies between your rewrite and the original paper.

Your task is to produce a NEW rewritten paper that fixes the identified issues while still following the structured format.

You will receive:
1. The raw LaTeX source of the proof (included below — ground truth — your rewrite must faithfully represent this).
2. Your previous rewritten proof, which contained faithfulness errors.
3. The specific discrepancies flagged by the checker.

Requirements for the new rewrite:
- Fix every issue listed in the Identified Errors. For each issue, make sure the new rewrite no longer deviates from the original paper in that way.
- Do NOT introduce new discrepancies: do not strengthen/weaken claims, omit steps, add arguments, drift in notation, or misrepresent the logical structure relative to the Original Paper.
- Preserve the original paper's content, notation, logical flow, ordering, and wording as much as possible.
- Follow the SAME structured output format (XML-style tags, multi-theorem schema, DEPS conventions, etc.) as the rewriting instructions below.
- **Use consecutive levels of decomposition starting from the top** (Theorem -> Proposition -> Lemma -> Claim -> Fact). Never skip an intermediate level. Match the tag to the dotted-id depth strictly: 1-dot id is a PROPOSITION, 2-dot is a LEMMA, 3-dot is a CLAIM, 4-dot is a FACT. Do not, for example, emit `<CLAIM_STATEMENT id="2.1">` (that would conflict with lemma id "2.1").

--- Original rewriting instructions (follow these for the output format) ---
{rewrite_instructions}
--- End of rewriting instructions ---

ORIGINAL LATEX SOURCE:
{original_paper}

PREVIOUS REWRITTEN PAPER (contains errors):
{previous_rewrite}

IDENTIFIED ERRORS:
{errors}

Now produce the corrected rewritten paper. Output ONLY the rewritten paper using the XML-tag format — no commentary before or after.
\end{promptbox}

\subsection{Block faithfulness check}
Checks each rewritten block against the original submission. It asks whether the
block claims more or less than the original, drops a load-bearing step, or adds
reasoning the submission does not contain.

\begin{promptbox}
You are an expert mathematician reviewing one block of a pseudo-formalized (PF) rewrite of a mathematical proof. You must check the block on TWO axes: FAITHFULNESS to the original proof, and structural HYGIENE.

You are given:
1. **Original Proof (LaTeX)**: The full LaTeX source of the original proof.
2. **Contexts**: Statements from the PF that the Assertion inherits definitions or assumptions from (e.g., the enclosing theorem statement, the enclosing proposition statement).
3. **Established Results**: Statements from the PF that have already been checked and can be assumed faithful. You may use them as reference points.
4. **Assertion**: The specific PF statement under review.
5. **Proposed Proof**: The PF proof of the Assertion (may be empty or "None").

FAITHFULNESS — flag any of (judge vs the Original Proof; do NOT judge whether the original mathematics is correct — if the original asserts something false or unjustified, a faithful block must assert it too):
- strengthened_or_weakened_claim: the Assertion claims more or less than the original establishes at the corresponding point.
- omitted_content: the Proposed Proof drops a non-trivial argument present in the original.
- added_content: the Proposed Proof introduces arguments, repairs, or justifications not present in the original.
- notation_drift: variables, functions, names, or definitions used differently than the original (including renaming or softening an author-given name).
- misinterpretation: the block misunderstands the original's reasoning or logical structure.
- scope_error: assumptions incorrectly inherited, dropped, or added.

HYGIENE — flag any of (judge the block on its own):
- proof_content_in_statement: the Assertion contains derivation/argument/commentary that belongs in the proof.
- circular_or_trivial_proof: the Proposed Proof merely restates the Assertion and contains no actual justification (NOTE: a bare assertion is CORRECT hygiene if the original also asserts the step without justification — that is faithfulness, not a defect).
- leftover_scaffolding: rewriter artifacts (e.g. "Original proof label: ...", duplicated "Assumptions / Conditions / Definitions" headers, document preamble like \title or \maketitle, disclaimers about the rewrite process itself). (NOTE: an "External citations: ..." list at the end of a proof block is NOT scaffolding — the PF format requires each proof block to list the external results it cites.)
- dangling_statement: the Assertion begins with a connective ("Then", "Therefore", "Hence", "Thus") whose antecedent is not in the block or its Contexts, so it does not stand alone.
- wrong_field: content generally placed in the wrong field.

Only flag genuine problems that alter mathematical meaning or block structure. Cosmetic rephrasing and terseness are fine.

At the very end of your response, output your verdict using XML tags. LaTeX may be used freely inside the tags — do NOT escape it. Do NOT write anything after the tags.

If the block is acceptable, end with exactly:
<verdict>PASS</verdict>

Otherwise end with:
<verdict>FAIL</verdict>
followed by one tag per problem found:
<issue category="one of the category names above">specific description: what the block says vs what it should say</issue>

**ORIGINAL PROOF (LATEX)**
{original_paper}

**CONTEXTS**
{contexts}

**ESTABLISHED RESULTS**
{established_results}

**ASSERTION**
{assertion}

**PROPOSED PROOF**
{proof}
\end{promptbox}

\subsection{Coverage check}
Complements the per-block check by asking whether the rewrite as a whole omits
material present in the original submission.

\begin{promptbox}
You are auditing whether a pseudo-formalisation (PF) of a math proof FAITHFULLY and COMPLETELY preserves the original author proof. This is a WHOLE-PROOF coverage check (not block by block).

Compare the FULL original author proof against the FULL reassembled PF below. Report every place where the PF, taken as a whole, DROPS, RENAMES, SANITISES, ADDS, or ALTERS content relative to the author — paying special attention to:
- definitions the author explicitly stated that are MISSING from the PF (e.g. an object the author defined but the PF only uses as if given);
- objects/notions whose author-given NAME was changed or softened (e.g. a specific possibly-nonstandard/nonexistent name replaced by a generic one);
- citations/attributions the author gave that the PF dropped, or that the PF added;
- hypotheses, quantifiers, or claims altered in strength (e.g. "almost every" vs "every", a dropped side-condition).
The PF FORMAT itself legitimately transforms the document; do NOT report any of the following as divergences:
- decomposition into THEOREM/PROPOSITION/LEMMA/CLAIM/FACT blocks with new numeric ids (author \label{{...}} names are dropped by design);
- reordering of material to eliminate forward references;
- dropping pure document furniture: title block, \maketitle, abstract-as-front-matter, section headings, and bibliography FORMATTING — provided the mathematical content and the attributions themselves survive somewhere in the PF;
- added structural scaffolding: DEPS tags, "None" proofs for steps the author merely asserts, statement/proof field separation.
Report ONLY divergences of mathematical or scholarly content (definitions, hypotheses, quantifiers, named objects, hedges/disclaimers, attributions).

For each finding: quote the author's text, quote (or note the absence in) the PF, and explain why it matters. If the PF is a faithful and complete rendering, say so explicitly.

At the very end of your response, output your verdict using XML tags. LaTeX may be used freely inside the tags — do NOT escape it. Do NOT write anything after the tags.

If the PF is faithful and complete, end with exactly:
<verdict>FAITHFUL</verdict>

Otherwise end with:
<verdict>HAS_DIVERGENCES</verdict>
followed by one tag per finding:
<finding>the author's text (quoted), what the PF does with it, and why it matters</finding>

== ORIGINAL AUTHOR PROOF ==
{raw}

== FULL REASSEMBLED PF ==
{pf}
\end{promptbox}

\subsection{Block verifier}

\begin{promptbox}
BLOCK VERIFIER — you have WEB SEARCH. You must be MAXIMALLY RIGOROUS about the mathematics: verify every load-bearing step, and give no benefit of the doubt to cited external results. Carry out ALL THREE steps before deciding.

STEP 1 — DEFINITION PINNING (object, not name). List every nontrivial term, object, operator, or named notion in the Assertion and its proof. For EACH, it is acceptable if:
  (a) it is defined in the proof itself, the Contexts, or the Established Results — an explicit formula, a defining role (e.g. a normalizing constant, a mollifier, the CDF of a stated density), or consistent unambiguous usage all count as definitions; OR
  (b) you retrieve a definition of the OBJECT from a credible source (textbook, peer-reviewed paper, established reference work) via web search. Quote the definition and cite the source.
OBJECT-PINNING RULE: pin the mathematical OBJECT, not its name. Notation or naming that differs from an external source's convention is NOT a defect. If the given material admits one coherent reading under which every use of the term is correct, adopt that reading and do not flag it. Return INCORRECT for an undefined or ambiguous term ONLY if the ambiguity infects the mathematics — some load-bearing step fails under EVERY reasonable reading consistent with the given material — and then name the offending term and the step it breaks. If the proof's usage genuinely conflicts with the definition the given material fixes for a term, that is also INCORRECT.

STEP 2 — LEMMA / CITED-RESULT PINNING (verbatim hypotheses AND conclusion). For every cited or invoked external result, retrieve its VERBATIM statement — BOTH its full hypotheses and its conclusion — from a credible source, and quote it. Then go through the hypotheses ONE ASSUMPTION AT A TIME and decide, for EACH SINGLE assumption, whether it genuinely holds in the current context — list the assumption, state holds/fails, and justify. Only if every single assumption fits may the result be applied. Also confirm that the conclusion the proof uses is exactly the conclusion the source states (not a stronger or broader version). Be especially alert to OVERGENERALIZATION (a result invoked in greater generality than the source actually establishes it) and to DOMAIN / OBJECT MISMATCH (the cited result is about a different class of objects, structure, or setting than the one at hand). A nonexistent, not-in-source, misstated, overgeneralized, mismatched, or otherwise misapplied result => INCORRECT, naming the exact failing assumption or the precise mismatch between what the source states and what the proof uses.

STEP 3 — SEEK A COUNTEREXAMPLE. Actively try to REFUTE the Assertion or a load-bearing claim in its proof, using whatever methods are appropriate. If you find a valid counterexample, return INCORRECT and state it explicitly.

Default-to-INCORRECT rule (cited results): if after honest effort you cannot confirm every hypothesis of a cited or invoked external result, do NOT give the benefit of the doubt — return INCORRECT and say which check failed. For definitions and terminology, apply the OBJECT-PINNING RULE above instead. Continue to allow routine, non-load-bearing steps to be filled: terseness, skipped arithmetic, standard manipulations, minor repairable slips, and standard background facts of the ambient theory that an expert referee would grant without citation (e.g. that a named contraction semigroup is a contraction, or that a normalizing constant normalizes) — provided the fact is not itself the crux of the Assertion.

ERROR REPORTING (when your verdict is INCORRECT): your error description is a referee report of DEFECTS ONLY, not an audit transcript. Do not narrate your procedure ("Step 1 ...", "Step 2 ..."), do not list terms you successfully pinned or checks that passed, and do not comment on parts of the block that are fine — do all of that work in your reasoning, and keep it out of the error description. Report EVERY independent defect you established, not only the first or most decisive one. State each as its own affirmative finding: the exact step it breaks, the specific failing object, hypothesis, or quantifier (e.g. "the cited result gives this only for \(\mu\)-a.e. x, while the proof needs it for every x"), and why it fails on the merits. 

Then apply the verification and output rules below. Do NOT consult First Proof materials (1stproof.org or its GitHub); rely only on independent literature.

You are an expert mathematical proof verifier specialized in research-level mathematics.

Your task is to verify whether the proposed proof of a specific statement, called "Assertion", is correct.

You are given:
1. **Contexts**: A sequence of statements from which the Assertion may or may not inherit definitions, assumptions, or conditions. These are often the parent or ancestor statements of the Assertion, and can be the same as the global theorem. They are provided solely so you can understand the definitions and assumptions of the Assertion. They have NOT been verified and may be incorrect. Do not treat them as established truths, and do not verify them yourself. Also do not automatically assume that the Assertion inherits assumptions or definitions from them. The Assertion will specify which settings or assumptions it inherits from these contextual statements.
2. **Established Results**: Statements that have already been verified or can be assumed to be correct. You may assume all established results are correct and use them freely — do NOT re-verify them. The proof of the Assertion can invoke these results as long as the assumptions are properly justified and the definitions are consistent.
3. **Assertion**: The specific statement whose proof you must verify.
4. **Proposed Proof**: The proof of the Assertion to verify.

Instructions:
- Verify ONLY the proposed proof of the Assertion.
- Read the Assertion carefully and analyze the proof step by step.
- Identify any incorrect, or logically invalid reasoning.

- When the proof references an established result, you may trust its conclusion, but you must verify that it is correctly applied:
    - Check that the result is used within its valid scope.
    - Explicitly identify the assumptions of the referenced result and confirm that each one is satisfied in the current context.
    - Verify that the definitions used in the invoked established results are the same as in the Assertion.
    - Detail which assumptions hold and why.
   - If the proof misapplies an established result, the error description must name which assumption failed to hold or which definition diverged from the Assertion's usage.
   - The proof is not required to restate the assumptions of a cited result. However, you must explicitly audit every use of a cited result: list all of its hypotheses and confirm, one by one, where each is satisfied in the current context. If any hypothesis is not actually satisfied, you must return INCORRECT for misapplication.

- Do NOT flag a step as incorrect merely because it omits intermediate justification, nor because it contains a local slip that a careful reader can repair on the spot.
Your job is to detect genuine errors — load-bearing false claims, misapplied results, logical invalidity, inconsistent use of terms — NOT to demand that every step be spelled out. Terseness, skipped arithmetic, standard manipulations, routine verifications, and minor slips are not errors when the intended correct statement is unambiguous from context and the rest of the proof still goes through; multiple such gaps do not compound into one.
  Before flagging, try to fill the gap or repair the slip yourself using the Contexts, the Established Results, and standard mathematical knowledge appropriate to the problem's level. Flag only when (a) the mistake is load-bearing (the Assertion's conclusion or a later step genuinely depends on the incorrect claim), or (b) the repair would require a substantive new idea, a nontrivial result, or a definition not available in the given material. In case (b), wrap each missing result in a <lemma> tag and each missing term in a <definition> tag — one tag per missing item:
  <lemma>full precise statement, including hypotheses and conclusion</lemma>
  <definition>term: full unambiguous definition</definition>
  A repair must not change the Assertion's hypotheses or conclusion: if fixing the proof would require adding an assumption, restricting the domain, or weakening the target, return INCORRECT.

- A Proposed Proof of exactly "None" is legitimate: treat it as an empty proof and apply the gap-filling test above — CORRECT if the Assertion is fillable from Contexts, Established Results, and standard knowledge; otherwise INCORRECT.

- When you DID fill gaps or repair minor slips yourself to reach CORRECT, record inside `<gap_filling>` what was missing or misstated and the reasoning used — enough that a reviewer could verify the step. Leave `<gap_filling>` empty when the proof was complete and error-free as written, or when the verdict is INCORRECT.

- Record your hypothesis audit of every cited result inside the `<cited_result_audits>` block — one `<audit>` entry per use of a cited result. Use this whenever the proof cites a result, even when the proof was complete as written and `<gap_filling>` is empty. Leave `<cited_result_audits>` empty ONLY when the proof cites no results at all. This requirement applies equally to CORRECT and INCORRECT verdicts.

At the very end of your response, you MUST output your final verdict using the tag format below. Do NOT write anything after the closing `</cited_result_audits>` tag. Inside any tag's text content you may write LaTeX freely — backslashes, braces, `<`, and `>` need NO escaping. The only requirement is that every opening tag has a matching closing tag exactly as shown.

If CORRECT, output:
<verdict>CORRECT</verdict>
<error_description></error_description>
<gap_filling>
<for each gap closed or slip repaired: what was missing or misstated and the reasoning used — concise but auditable; leave empty if the proof was complete and error-free as written>
</gap_filling>
<cited_result_audits>
<audit>
<cited_as><exactly what the proof wrote></cited_as>
<hypothesis>
<statement><the cited result's hypothesis></statement>
<satisfied>true</satisfied>
<justification><concrete reason it holds in the current context></justification>
</hypothesis>
</audit>
</cited_result_audits>

If INCORRECT, output:
<verdict>INCORRECT</verdict>
<error_description>
Identify the specific step that fails, state what it claims, and explain why it is wrong or unjustified.
</error_description>
<gap_filling></gap_filling>
<cited_result_audits>
<audit>
<cited_as><exactly what the proof wrote></cited_as>
<hypothesis>
<statement><the cited result's hypothesis></statement>
<satisfied>true_or_false</satisfied>
<justification><concrete reason it holds or fails in the current context></justification>
</hypothesis>
</audit>
</cited_result_audits>

**CONTEXTS**

{contexts}

**ESTABLISHED RESULTS**

{established_results}

**ASSERTION**

{assertion}

**PROPOSED PROOF**

{proof}
\end{promptbox}

\subsection{Calibrator}
A second-pass referee over the block verifier's flags, adapted from the arXiv
meta-verification step.

\begin{promptbox}
You are an expert mathematical referee. Your task is to produce the FINAL list of mathematical errors in a submitted mathematical proof.

You are given:
1. The ORIGINAL proof exactly as submitted. This is your authoritative source for the proof's content — precise notation, equations, exact wording.
2. A structured rewrite of the proof, decomposed into theorems / propositions / lemmas. The rewrite was produced by an automated rewriter and may itself introduce mistakes, omissions, or distortions that were **not** present in the original proof.
3. A list of errors that an automated component-verifier flagged in specific components (lemmas, propositions, or theorems) of the rewritten proof. These potential errors may or may not be genuine errors in the underlying proof — in particular, an "error" may be an artifact of the rewriting process (e.g., the rewriter dropped a key step, misstated a claim, or restructured the argument in a way that obscures correct reasoning that **is** present in the original proof).

Evaluation process:

1. **Error validation.** For each potential error, carefully determine whether it is a genuine error in the underlying proof or a false alarm. Examine the error in the context of the full rewritten proof, AND cross-check against the original proof to see whether the reasoning the rewritten proof is missing or misstating actually appears (correctly) in the original.

A flag can fail for exactly one reason:
  - `rewrite_artifact` — the reasoning the flag says is missing or misstated IS present and correct in the ORIGINAL proof. To use this verdict you MUST quote the passage of the original proof that supplies it. If you cannot locate such a passage, this verdict does not apply.

You may NOT reject a flag on the grounds that the component-verifier's mathematical objection is itself mistaken. The component-verifier examined the component in depth and audited the results it cites; second-guessing its mathematics is out of scope for you. Your ONLY question is whether the ORIGINAL proof already supplies what the flag says is missing or misstated.

If you cannot do the above then the default is to accept the error as `genuine`.

2. **Report only mathematical errors.** Focus on errors of mathematical substance (incorrect proofs, unjustified steps, false claims, gaps in reasoning, miscomputed quantities, misapplied theorems, not well defined objects). 

Report a verdict for EVERY potential error in the list, in the order given. Do not introduce errors in components that were not flagged.

OUTPUT FORMAT:

Output your final answer using the following XML-style format. Emit one `<error>` block for EVERY potential error you were given, in the order given, setting `<component>` to that potential error's component label exactly as given to you below. `<verdict>` must be exactly one of `genuine` or `rewrite_artifact`. `<reason>` is one or two sentences justifying that verdict — for `rewrite_artifact` it must quote the passage of the ORIGINAL proof that supplies the reasoning. Include `<description>` only when the verdict is `genuine`. Use the field tags exactly as shown. You may write LaTeX math (with raw backslashes) freely inside the `<reason>` and `<description>` fields; do not escape anything. Do not output anything after the closing `</errors>` tag.

<errors>
  <error>
    <component>Lemma 2.3</component>
    <verdict>genuine</verdict>
    <reason>Why this is a real error in the proof as submitted.</reason>
    <description>Brief description of the mathematical error and why it is wrong.</description>
  </error>
  <error>
    <component>Proposition 4</component>
    <verdict>rewrite_artifact</verdict>
    <reason>The original proof supplies this step: "<quote from the original>".</reason>
  </error>
</errors>

ORIGINAL PROOF AS SUBMITTED:
{original_proof}

REWRITTEN PROOF (PROPOSED STRUCTURED FORM):
{rewritten_proof}

POTENTIAL ERRORS (from automated verification of the rewritten proof):
{errors}
\end{promptbox}

\subsection{\texttt{FirstProofGradingBench} verifier auto eval}
\label{sec:verifier-auto-eval}
 The verifier auto eval 
 is an LLM judge that 
 receives the submitted proof,
the ground-truth errors, and the predicted errors, and reports for each
prediction which ground-truth errors it matches.

We use this system because 
human evaluation is not
scalable. To test 
the reliability of the auto
eval system, we 
analyzed eight of the 
error predictions 
generated by the PF 
verifier that the 
verifier auto eval
labeled as true positives. 
Six of these we agreed 
with the judgment 
of the verifier auto eval
(03B and 03C each had 
two predictions that were 
correct, 04C had a prediction 
that was correct although 
contained other superfluous information). 
Two of the eight we deemed 
an incorrect labeling by the auto eval
(05C and 08A).

We use GPT 5.5 
with the following prompt:

\begin{promptbox}
You are adjudicating whether an automated verifier found the KNOWN errors in a mathematical proof.

You are given:
1. The original proof (LaTeX).
2. GROUND-TRUTH ERRORS: the real errors in this proof, established by expert referees (each: a verbatim excerpt containing the error, and a description). Treat these as correct and complete.
3. PREDICTED ERRORS: errors reported by an automated verifier (each: a verbatim quote, an optional location note, and a description).

For EACH predicted error, decide which ground-truth error(s) it matches. To decide if a predicted error matches a grount truth error, take the following steps.
FIRST, distil the CONCRETE ESSENCE of the ground truth error: the single specific, load-bearing defect — e.g. the exact false statement, the specific hypothesis / side-condition that is missing or violated, the specific object that is undefined / nonexistent / misdefined, the specific quantifier that is overstated (e.g. "a.e." used as "everywhere"), the specific cited result that does not exist or is misapplied, or the concrete counterexample. State this essence in one sentence.
THEN judge the predicted error against THAT essence. The bar for "agree" is that the verifier independently puts its finger on the SAME concrete essence — the same specific mechanism, object, condition, quantifier, or citation — so that an expert reading the predicted errors reasoning would conclude it found the very same defect, not merely that it was suspicious of the right region.


Briefly analyze each predicted error against the ground-truth list — one short paragraph per prediction, stating which ground-truth error it matches (if any) and why. Then, at the very end of your response, output the verdict as a JSON object on its own line — nothing after it:
{{"pred_matches": [<one entry per predicted error, in order: a list of 0-based ground-truth indices it matches, empty if none>]}}

**ORIGINAL PROOF (LATEX)**
{proof}

**GROUND-TRUTH ERRORS**
{gold_list}

**PREDICTED ERRORS**
{pred_list}
\end{promptbox}

\section{Compute cost}
\label{app:compute-cost}

This appendix records the token and dollar cost accounting for our
verification runs.  Costs use standard GPT-5.4 mini pricing:
\$0.75/M input tokens, \$0.075/M cached input tokens, and
\$4.50/M output tokens.  For both benchmarks, the direct baseline
makes one grading call per rollout, while PF runs the full
pseudo-formalization, faithfulness, block-verification, and final
calibration pipeline once per rollout.  \Cref{tab:cost-accounting}
reports aggregate token usage and API cost for the main runs. We note
that while our algorithm has higher cost per round, the Pareto
improvement when we sweep $k$ in~\Cref{fig:figure_1}
cannot be attained by simply scaling inference-time compute.

\begin{table}[h]
    \centering
    \footnotesize
    \setlength{\tabcolsep}{3.5pt}
    \caption{Token usage and API cost for the benchmark runs.}
    \vspace{3mm}
    \label{tab:cost-accounting}
    \begin{tabular}{@{}lrrrrrr@{}}
        \toprule
        Setup & Papers $\times$ runs & Per-rollout in & Per-rollout out
        & Total in & Total out & Cost \\
        \midrule
        Hard2Verify baseline & $200 \times 8$ & 2.6K & 3.9K & 4.18M & 6.19M & \$30.07 \\
        Hard2Verify PF       & $200 \times 8$ & 64K  & 26K  & 102.41M & 42.03M & \$264.24 \\
        Arxiv baseline       & $100 \times 8$ & 24K  & 17K  & 19.04M & 13.63M & \$75.60 \\
        Arxiv PF             & $100 \times 8$ & 298K & 165K & 238.40M & 131.86M & \$772.17 \\
        \bottomrule
    \end{tabular}
\end{table}

The higher cost of PF is distributed across the phases of the algorithm
rather than concentrated in a single call.  To make this visible,
\Cref{tab:hard2verify-pf-stage-cost} gives phase-level accounting for
the PF run on \texttt{Hard2Verify}.

\begin{table}[h]
    \centering
    \small
    \caption{
        Stage-level token usage and API cost for \texttt{Hard2Verify}
        PF/BV plus v1 calibration at $k=8$ over 200 proofs.
    }
    \vspace{3mm}
    \label{tab:hard2verify-pf-stage-cost}
    \begin{tabular}{lrrrrr}
        \toprule
        Stage & Calls & Input & Cached in & Output & Cost \\
        \midrule
        Rewrite & 1,600 & 6.12M & 0.00M & 13.68M & \$66.14 \\
        Faithfulness / parse checks & 15,826 & 58.05M & 0.00M & 5.88M & \$70.01 \\
        Block verification & 15,857 & 27.12M & 0.00M & 15.86M & \$91.70 \\
        Step-level calibration (v1) & 1,600 & 11.12M & 2.55M & 6.62M & \$36.40 \\
        \midrule
        Total & 34,883 & 102.41M & 2.55M & 42.03M & \$264.24 \\
        \bottomrule
    \end{tabular}
\end{table}

\end{document}